\documentclass[12pt]{article}
\usepackage{amsmath,amssymb,amsthm,amsfonts,setspace,graphicx,xcolor,mathabx,paralist,paralist,nicefrac,comment, hyperref,booktabs,multirow, float, caption, thmtools,thm-restate,caption,bbm,appendix,titlesec, array, tikz, adjustbox,subcaption}
\usepackage{subcaption}
\usepackage[multiple]{footmisc}

\usepackage[style=authoryear-comp,natbib,maxbibnames=9,maxcitenames=3,uniquename = false, uniquelist=false,
    backend=biber]{biblatex}
\renewbibmacro{in:}{}

\usepackage[top=1.25in,bottom=1.25in,left=1.25in,right=1.25in]{geometry}

\hypersetup{colorlinks=true,linkbordercolor=red,linkcolor=blue,citecolor=blue,pdfborderstyle=none}

\DeclareMathOperator*{\argmax}{argmax}
\DeclareMathOperator*{\interior}{int}

\newtheorem*{thm3.6_BHS}{Theorem 3.6 (BHS)}
\newtheorem*{thm2.2_Schreiber}{Theorem 2.2 (\textcite{schreiber2001urn})}
\newtheorem*{Assumptions_BHS}{BHS Standing Assumptions}
\newtheorem*{thm2.9_Pemantle}{Theorem 2.9 (\textcite{pemantle2007survey})}
\newtheorem{theorem}{Theorem}

\newtheorem{lemma}{Lemma}

\newtheorem*{cor-pur}{Corollary 4}
\newtheorem*{lem-tech}{Lemma 1}
\newtheorem*{lem-tech-DF}{Lemma 7}

\newtheorem*{assumption*}{Assumption}
\theoremstyle{definition}
\newtheorem{definition}{Definition}
\newtheorem{assumption}{Assumption}

\DeclareMathOperator{\sign}{sign}

\renewcommand{\blacksquare}{\rule{0.5em}{0.5em}}
\renewenvironment{proof}[1][Proof]{\medskip\noindent\textbf{#1.} }{\hfill\blacksquare\par\medskip}

\title{Endogenous Attention and the Spread of False News (Extended Cut)\footnote{This an extended version of the paper with additional details. We thank Michel Benaim, Claire Bartolone, Yifan Dai, Krishna Dasaratha, Ben Golub, Kevin He, Navin Kartik, Rachel Kranton, Simon Loertscher, David McAdams, Reed Orchinik,
 David Rand, Doron Ravid,  Noah Siderhurst,  Philipp Strack, and Alexander Wolitzky for helpful comments and conversations. We thank NSF Grant SES-2417162 for financial support. }}
\author{Tuval Danenberg\thanks{Department of Economics, MIT, Cambridge, MA, 02142, tuvaldan@mit.edu}\, and Drew Fudenberg\thanks{Department of Economics, MIT, Cambridge, MA, 02142, drew.fudenberg@gmail.com} }

\date{\today}

\begin{document}

\maketitle
\thispagestyle{empty}
\vspace{-1em}

\begin{abstract}
 We study the impact of endogenous attention in a dynamic social media model. Each period, a user observes a random story and decides whether to share it. Users like sharing true and interesting stories, but identifying false stories requires costly attention.  
Depending on parameters, the system exhibits either a unique limit or strong path dependence. 
Endogenous attention responds to changes in false story credibility, so reducing credibility can boost their prevalence. Increases in the exogenous production rate of false stories can be amplified by users' sharing decisions. Increasing users' capacity to reach others amplifies both true and false stories; we identify conditions under which the net effect favors truth over falsehood.

\quad



\end{abstract}

\newpage

\setcounter{page}{1}

\section{Introduction}
  This paper develops a dynamic model of the spread of misinformation on social media. \textcite{vosoughi2018spread} shows that the spread of falsehoods on social media  is mostly due to humans rather than bots, and  \textcite{pennycook2021shifting} attributes the sharing of false news to inattention. 
   Motivated by these  findings, our  model  assumes that users want to share true stories, but  distinguishing true and false content requires costly attention.  Users' attention depends on the prevalence and credibility of false stories:  If the share of false stories in their feed is negligible, they are unwilling to spend much effort spotting them, but if that share is significant and the false stories are superficially plausible, users  are willing to incur a nontrivial  
  cost to distinguish between true and false content. In turn, attention choices affect the prevalence of false stories as more attentive users are more discerning. We study the resulting joint dynamics of users' attention and platform composition.

In our model, each period a distinct user randomly draws a story from the stories on a social media platform and decides whether to share it. Users consider two factors when evaluating a story: its \emph{veracity}, or truthfulness, and their  subjective \emph{interest} in it. Users first observe the story's interest level, and then choose  their  attention level and pay the corresponding cost. They then receive a binary signal of the story's veracity. One signal realization reveals that the story is false while the other may be sent for both true and false stories.  False stories are characterized by a credibility measure that captures how true they appear---when false stories are highly credible, signals about their veracity are less precise. The precision of the signal increases with the user's attention level, and is supermodular in credibility and attention, so  users pay more attention when false stories are more credible. If the user decides to share the story, a fixed number of identical copies are added to the platform. Regardless of the sharing decision,  fixed numbers of true and false stories are exogenously added as well, which corresponds to original content creation.

 We assume that  users do not share boring stories, and consider two levels of interest: mildly interesting (M) and very interesting (I). While a story's veracity is fixed throughout time, each user's interest in it is drawn i.i.d (conditional on veracity). This captures the idea that different users will find different stories very interesting. We also assume that false stories are more likely to be very interesting. 

Our focus is the share of true stories in the system for each period $n\in\mathbb{N}$, which we denote by $y_n$.  As a function of $y_n$, users follow one of four possible decision rules. For each interest level, the rules dictate one of two approaches:  Either  i) do not pay attention to  stories of that level and do not share; or ii) pay attention to those stories  and share them as long as they were not revealed to be false. When users do pay attention to a story, their attention level varies with  $y_n$ according to a first-order condition that equates the marginal cost of attention with its marginal benefit.  When $y_n$ is sufficiently high, the system is in the \emph{sharing} region, where users pay attention to stories of both interest levels. When $y_n$ is low, the system is in the \emph{no-sharing} region, where users do not pay attention to any story. In between, there is an intermediate region, in which the optimal decision rule is either to only pay attention to very interesting stories or  only to mildly interesting ones.

 By applying stochastic approximation techniques to concatenations of generalized Polya urns,  we 
show that $y_n$ converges almost surely and provide a complete characterization of its limit. (See the technical summary below.) For some parameter values the limit is unique. For others it is random, so that starting from the same initial conditions the platform may end up with significantly different limit shares of true stories and different user behavior in the limit. This effect is most pronounced when the platform is new and the total number of stories is small, but it is still present in any finite platform.

The system converges either to a point where users strictly prefer a single decision  rule or to one where they are indifferent between two rules. Comparative statics are qualitatively different in these two cases.\footnote{This is analogous to the difference in comparative statics between pure-strategy and mixed-strategy Nash equilibrium in games.} For example, in the steady states where users strictly prefer one rule, increasing the cost of attention lowers the share of true stories. However, 
in the steady states where users are indifferent between two  rules, the limit share of true stories is increasing in the cost of attention. This occurs because the cost of attention lowers  users' payoffs, so 
the share of true stories required for indifference increases.   

The most notable  implications  of our comparative statics analysis concern false-story credibility, user reach, and the production rate of false stories.   False story credibility can have a  non-monotonic effect, because when false stories are very credible users  pay more attention to them. When credibility is high, user responses to an increase in credibility may more than compensate for the direct effect of this increase, so  making false stories less credible can actually increase their prevalence. We draw two practical implications. First, producers of false stories may choose to produce implausible stories even when credibility is free. Second, platforms that aim to counter the spread of false news by fact-checking false stories might be better off not fact-checking at all than fact-checking only a small share of stories, because increasing the share of stories flagged as false leads users to put more trust in stories that were not flagged.

Turning to the effects of users' reach, a common view is that social media platforms are hotbeds for false news because users can easily disseminate content to large audiences. However, the ability of each user to reach many others can also increase the spread of true stories. To analyze this trade-off, we formally define \emph{reach} as the number of copies of a story added to the platform when a user shares it. We find that the effects of increased reach depend on users' \emph{truth-sharing propensity} and  \emph{false-sharing propensity}, defined as the probability of sharing a true or false story, respectively, when one is drawn. We  say that users are \emph{discerning} if their truth-sharing propensity exceeds their false-sharing propensity.  We show that whether users are discerning determines whether sharing increases or decreases the share of true stories. This supports the use of similar measures in \textcite{pennycook2022accuracy} and \textcite{guriev2023curtailing} to evaluate attempts to reduce the spread of false news.

When users only share mildly interesting stories or share both mildly and very interesting stories, they are always discerning. This implies that the steady states associated with these decision rules are increasing in the reach parameter, and converge to $1$ as the reach parameter increases to infinity. In contrast, when users share only very interesting stories (decision rule $I$), they may or may not be discerning, depending on the parameters. Intuitively, because users have a higher intrinsic benefit from sharing very interesting stories, they may share stories that are relatively likely to be false. 

If users are not discerning, then their net effect on platform composition is negative. In this case, higher reach is detrimental, degrading the quality of information on the platform.  However, this negative  effect of reach is self-limiting: as reach  increases, the decrease in the steady-state share of true stories makes users more attentive and therefore more discerning, which partially offsets the direct negative effect.  So, while for some parameter specifications the steady state for rule $I$ converges to $0$ as reach increases, for others it converges to an interior point that equates users' truth-sharing  and false-sharing propensities. The case where users are discerning with decision rule $I$ is a mirror image. 

Finally, we find that increases in the production rate of false stories can cause users to stop sharing entirely, amplifying the exogenous shock. Yet in other cases, a significant share of false stories can persist in the limit even as their exogenous inflow vanishes.

\subsection*{Technical Summary}
In the Polya urn model, an urn consists of balls of various colors. In each period one ball is drawn randomly from the urn, and then  returned to the urn along with one additional ball of the same color. A \emph{generalized Polya urn}  (GPU) allows for the number of balls added in each period to be random, with probabilities that depend on the state of the system; see, e.g., \textcite{schreiber2001urn} and \textcite{mahmoud2008polya}.\footnote{\textcite{arthur1993information},  \textcite{smith2020rational}, and 
 and \textcite{arieli2024sequential} use  GPUs with two colors and one ball added each period to analyze models of social learning. }

In our model, if the users' decision rule was fixed, instead of depending on $y_n$, our system would be a GPU  where stories are balls, colors are veracity levels, and the number of balls added each period depends on whether the user decides to share their current story and whether that story was true or false. 
\textcite{schreiber2001urn} and \textcite{benaim2004generalized} use stochastic approximation arguments to show that under fairly general conditions the long-run behavior of GPUs can be determined by studying the attractors of a deterministic differential equation.  Their results imply that the hypothetical systems where users pick one of the four contingently-optimal decision rules and use it for all values of $y_n$  have unique limit shares of true stories.

These limits, which we call \emph{quasi steady states}, are the unique steady states of the associated differential equations. However, because the optimal sharing and attention rules are not continuous,  our system is not a GPU but a concatenation of them. For this reason, we extend the literature on stochastic approximation of urn models to cover concatenations of a finite number of GPUs. This lets us relate the long-run behavior of the system to the stable steady states of  the associated   \emph{limit differential inclusion} (LDI), which concatenates the   differential equations associated with the GPUs.\footnote{A differential inclusion is an equation of the form $\frac{dx}{dt}\in F(x)$ for a set-valued function $F$.}

The first step in our analysis of the dynamics of the  share of true stories is Theorem \ref{thm S_F}, which  shows that a quasi steady state is a stable steady state for the LDI if and only if it is within the region where its associated decision rule is uniquely optimal. Depending on the parameters, there may or may not be one additional stable steady state, the \emph{threshold} where the user is just indifferent between sharing and not sharing very interesting stories.

Next, Theorem \ref{thm limit chain transitive} in Appendix \ref{appendix:stochastic approximation} uses results from \citet{benaim2005stochastic} (henceforth BHS) to show the system almost surely converges to a steady state of the
 LDI.  Lemma  \ref{claim stable} then  gives a direct proof that all of the stable steady states of the limit differential inclusion have positive probability. Lemma \ref{claim repelling} complements this by using a result of \citet{pemantle2007survey} to show the system has  probability $0$ of converging to an unstable steady state. Together these results imply Theorem \ref{thm limit behavior}, which shows that the system almost surely converges to a stable steady state of the limit differential inclusion, and converges to each stable steady state with positive probability (except in the case where users never share and the system deterministically converges to the no-sharing steady state).

\section{Related Literature}\label{sec:literature}

\subsubsection*{Empirical Evidence} 
In our model, inattention plays a central role in the sharing of false content. \textcite{pennycook2021shifting} argues  that inattention to  veracity is a key reason that users share false stories. Using survey and field experiments on Twitter (now X), it  demonstrates that there is a discrepancy  between users' stated preferences for accuracy and their observed sharing behavior, and that  interventions designed to shift users’ attention toward accuracy significantly improve the quality of the content they subsequently share.

We assume that users care about two content dimensions---veracity and interest. \textcite{chen2023makes} conducts a factor analysis of the content dimensions affecting sharing decisions in a series of experiments and finds that the main factors are perceived accuracy, evocativeness, and familiarity, and that  accuracy   has the most impact on sharing. The first two factors correspond to veracity and interest, with the qualification that interest in our model is subjective.

\subsubsection*{Theory of Online Misinformation}

Much of the theoretical literature on misinformation studies sequential models in which a single  story spreads through a network.  In \textcite{bloch2018rumors}, the network consists of biased and unbiased agents, and agents' evaluations of the story's veracity depend on their beliefs about the part of the network in which the story originated. In \textcite{acemoglu2023online}, the social media platform chooses the network structure to maximize engagement, and
users' sharing decisions depend on expectations about subsequent users' tastes but not past user actions. These models do not feature inspection choices or endogenous attention. In \textcite{papanastasiou2020fake},  users arranged in a line can pay a cost to learn a story’s veracity. Inspection choices depend on a user’s position in the line (which signals prior inspections) and an exogenous belief about the prevalence of false news.

Other papers have studied the spread of multiple messages about a fixed binary state. In \citet{mostagir2022naive}, each user initially gets an informative message about the state, and then repeatedly transmits their posteriors to their neighbors using either Bayesian updating or DeGroot learning. \citet{merlino2023debunking} analyzes the mean field of an infinite-population  SIS model with two messages corresponding to the two states.  Agents become ``infected'' when they encounter a message and choose how much effort to spend to verify it, so this model features a form of endogenous attention. \textcite{kranton2024social} models the production of information, with consumers allocating attention across sources and veracity endogenously determined by producer investments.

\textcite{dasaratha2022learning}, like our paper, uses stochastic approximation to determine the evolution of the shares of true and false stories rather than the spread of a single story. Users  only care about veracity and do not know the state of the platform. The paper focuses on  the weight the platform places on stories' virality when choosing what stories to display to users, and does not feature endogenous attention.

We contribute to this literature by considering agents whose sharing decisions depend on both the individual characteristics of a story and the endogenous composition of the platform. In existing models, the prevalence of false news is either fixed or inconsequential to users. In our model, the evolving platform composition feeds back into users' attention and sharing decisions. Our extension of generalized Polya urn results to concatenations of urns allows us to analyze sharing dynamics in a setting where some platform compositions lead users to completely ``tune out'' from considering certain types of stories. 
The results highlight the role of path dependence in platform evolution and clarify how different determinants of users' sharing behavior shape long-run outcomes. As we discuss below, our result about the nonmonotone effect of false story credibility is related to previous findings in \textcite{pennycook2020implied} and \textcite{acemoglu2023online}.

\section{Model}\label{sec:model}

We consider an infinite-horizon model of a social media platform. Each story has a fixed \emph{veracity} $v\in\{T,F\}$, with the story being true if $v=T$ and false otherwise. When a user draws a story they observe their  \emph{interest level} $\ell$, which is either $M$ (mildly interesting) or $I$ (very interesting).\footnote{In reality there are also boring stories that are rarely or never shared, we omit these.} The probabilities of each interest level are: 
\begin{equation*}
    \mathbb{P}(\ell=I|v=T)=\frac{1}{2};\, \mathbb{P}(\ell=I|v=F)=\delta.
\end{equation*}
 We assume that $\delta > \frac{1}{2}$, so  false stories are more likely to seem very interesting, and that  $\delta<1$ as otherwise mildly interesting stories are always true. 
 
 The false stories are of \emph{credibility} $\theta \in (0,1).$ The credibility of a false story determines how difficult it is to distinguish from a true story, in a manner that will be described below. To keep the model simple, we assume that all false stories have the same credibility. 
 
 The platform begins operating at time $n=0$ with an initial stock of true and false stories $(T_{0}, F_{0}).$ 
 Let $T_n$ and $F_n$ respectively denote the numbers of true and false stories on the platform at the beginning of period $n$. The vector $z_n:= (T_n,F_n)$ summarizes the current state of the platform;  we use the notation $|z_n|:= T_n+F_n$ for the total number of stories in period $n$, and let $y_n := \frac{T_n}{|z_n|}$ denote the share of true stories. 

In each period $n\ge0$, a distinct user randomly draws a story among those currently on the platform and decides whether or not to share it. Before making the sharing decision, the user observes $\ell$ and a noisy signal of $v$. The precision of this signal depends on the user's \emph{attention} as will be explained below. The parameter $\rho $ describes the \emph{reach} of shared stories on the platform---if the user decides to share the story, $\rho$ copies of the story are added to the platform. Regardless of the user's decision, one true story and $\kappa$ false stories are posted to the platform, corresponding to original content creation.

In summary, in each period a distinct user does the following in sequence:
\begin{enumerate}
\item   Draws a story and observes $\ell$. 
    \item Chooses an attention level $a\in[0,1]$.
    \item Draws a signal whose distribution depends on $a$.
    \item Decides whether to share the story. If shared, $\rho$ copies of the story are added.
    \item Receives payoffs.
      \end{enumerate} 
   Then at the end of the period  $1$ new true story and $\kappa$ new false stories are exogenously added.

 The cost of attention level $a$ is $\beta\cdot a^2$, where $\beta>0$. The signal is $s\in\{T',F'\}$, with  probabilities given by 
\begin{equation}\label{signals}
  \mathbb{P}(T'|T) = 1 ;\, \mathbb{P}(T'|F) = \theta(1-a). 
\end{equation}
The idea behind Equation \ref{signals} is that a false story of credibility $\theta$ is \textit{clearly false} with probability $1-\theta$, where a  clearly false story is one that users will recognize as false even when they do not pay attention. With probability $\theta$,  users will notice the story is false only if they pay attention.  A user's attention level $a$ is the probability with which they pay attention. Thus, when a user's attention level is $a$ and the credibility of false stories is $\theta$, they will identify a false story as false with probability $\mathbb{P}(F'|F) = 1-\theta+\theta a = 1-\theta(1-a)$. If the story is true,  the user receives the signal $T'$ with certainty, regardless of  their  attention level. Thus, signal $F'$ reveals the story is false, while after  signal $T'$ the user is uncertain about the story's veracity.

Users choose their  attention level after observing their interest in the story, knowing the current share of true stories $y_n$.\footnote{This approximates a scenario where the veracity of stories shared a few periods back has been revealed and the mix between true and false stories is not changing too quickly.} They will never share stories for which they received the signal $F'$, so they either share stories with signal $T'$ or do not share at all. Whether or not they share, users pay the cost $\beta a^2$ of their chosen attention level. If they do not share they get no additional payoff so their total payoff is $-\beta a^2$. If they share a $(v,\ell)$-story their additional payoff is
\[
u(v,\ell) = 1-\mu\mathbbm{1}(v=F)+\lambda\mathbbm{1}(\ell=I).\]
 Here we have normalized the payoff to sharing a true and mildly interesting story to $1$. The parameter $\mu$ captures the loss from sharing a false story, and the parameter $\lambda$ captures the additional gain from sharing a story that is very interesting. We assume both of these are strictly positive, so in line with the empirical evidence mentioned above, users want to share stories that are true and interesting. We also make the following parametric assumption:
 \begin{assumption}\label{condition on lambda}
$\mu >1+\lambda$.
\end{assumption}
Assumption \ref{condition on lambda} implies users will not share very interesting stories they know are false, and therefore will not share any story for which they received the signal $F'$.\footnote{This assumption is consistent with \textcite{chen2023makes}, which finds that the content factor with the strongest positive correlation with sharing intentions is perceived accuracy.} Hence, for each $\ell$, users either pay no attention and do not share, or they pay attention and share if and only if they receive the signal $T'$. In the latter case, their expected payoff to attention level $a$ is 
\begin{equation}\label{definition of U}
     U(a,y,\ell) := \mathbb{P}_{a,y}(T'|\ell)\mathbb{E}[u(v,\ell)|T',\ell]-\beta a^2.
 \end{equation}
Thus,  if users share after $T'$ they will choose the attention level
\[
a(y,\ell) :=\argmax_{a\in[0,1]} U(a,y,\ell).\]
 
Our second parametric assumption implies that the optimal attention levels are always given by solutions to first-order conditions.

\begin{assumption} \label{sufficient for a below 1} 
$(\mu-1)\theta <2\beta$.
\end{assumption}

In summary, the model parameters are $(\rho,\kappa,\theta,\mu,\beta,\delta,\lambda)$. We assume throughout that all parameters are strictly positive, satisfy Assumptions \ref{condition on lambda} and \ref{sufficient for a below 1}, and that $\theta<1$ and $\delta\in(\frac{1}{2},1)$. 

\subsection*{Discussion}

Users in our model want to share content that is both subjectively interesting and objectively true. Each of these motivations can be interpreted either as intrinsic or as a reduced form approximation of a motive that depends on subsequent users' reactions. For example, users might share subjectively interesting content simply due to a \emph{self-expression} motive.  Alternatively, they could share it in the hope of influencing others to take actions aligned with their interests, similar to the persuasion motive in \textcite{guriev2023curtailing}. Formally incorporating these latter considerations would significantly complicate our model. Instead, we focus on  how these empirically grounded motives, previously derived in simpler models, interact with endogenous attention on a dynamically evolving platform.

\section{Optimal Attention and the Sharing Decision}\label{subsec:sharing decision}
 
 We are interested in characterizing the composition of stories on the platform over time, i.e, analyzing the stochastic process $\{z_n\}$, and in particular the share of true stories $\{y_n\}$. To begin the analysis, we compute how user-optimal attention depends on the  state.

\begin{lemma}\label{lemma U a V}

The functions $U(a,y,M)$ and $U(a,y,I)$ are strictly concave, and the optimal attention levels (conditional on sharing $T'$ stories) are:

\begin{equation*}\label{attention levels}
\left\{ 
\begin{array}{ll} 
 0\le a(y,M) = \dfrac{(\mu-1)(1-y)(1-\delta) \theta}{\beta\left(y+2(1-y)(1-\delta)\right)}\le1,
\\ 
\\
0\le a(y,I) = \dfrac{(\mu-1-\lambda)(1-y)\delta \theta}{\beta\left(y+2(1-y)\delta\right)}\le 1.
\end{array}
\right.
\end{equation*}
\end{lemma}
The proof of this and all other results stated in the text are in Appendix \ref{appendix:proofs}. 
It is straightforward to verify that $a(y,\ell)<1$ for all $y$ and $a(y,\ell)>0$ if $y<1$, and that  the system can never reach a state where $y=1$. Intuitively, when $y=1$ there is no need to pay attention, so  $a(1,M)=a(1,I)=0$. As $y$ decreases the marginal gain from paying attention increases,  and since the $U$'s are strictly concave, $da/dy<0.$ However, when $y$ is close enough to $0$ the payoff from attention is so low that users do  not share any stories and do not pay attention at all. 
 
Online Appendix \ref{appendix:additional} shows that
users pay more attention to stories of both interest levels when false stories are very credible and when the cost to sharing false stories is high. Users pay more attention to very interesting stories  when false stories are more likely to be very interesting, and pay less attention to very interesting stories as the payoff to sharing them increases. 

The next lemma shows that there are interior \emph{thresholds} $\hat{y}_M,\hat{y}_I$  for each interest level such that if the share of true stories is below the corresponding threshold then users choose $a=0$ and do not share the story, and if the share is above this threshold users choose the attention level given in Lemma \ref{lemma U a V} and share if and only if they received the signal $T'$.

\begin{lemma}\label{lemma thresholds}
    Let $V(y,\ell) := U(a(y,\ell),y,\ell)$. $V(y,M)$ and $V(y,I)$ are strictly increasing in $y$, and there are (unique) $\hat{y}_M,\hat{y}_I\in(0,1)$ such that  $V(\hat{y}_M,M)=V(\hat{y}_I,I)=0$.
\end{lemma}

\begin{table}[h]
\centering
\caption{\textbf{Regions and Decision Rules}}
\label{table:regions}
\setlength{\tabcolsep}{13pt}
\begin{tabular}{ll@{}}
\toprule
$N= (0,\min\{\hat{y}_M, \hat{y}_I\})$& Don't share any story.\\
$I=(\hat{y}_I,\hat{y}_M)$ &  Share only very interesting (with signal $T'$).\\
$M =(\hat{y}_M,\hat{y}_I)$ & Share only mildly interesting (with signal $T'$).\\
$S= (\max\{\hat{y}_M, \hat{y}_I\},1)$ & Share both mildly and very interesting (with signal $T'$).\\
\bottomrule
\end{tabular}
\end{table}

Users' sharing behavior depends on the share of true stories $y_n$. When $y_n$ is below both thresholds, the expected value from sharing is negative for both interest levels so users do not share at all. When $y_n$ is above both thresholds, users share both types of stories,  and otherwise they share only one type of story, as shown in Table \ref{table:regions}. (When the state is exactly at a threshold, users are indifferent between the two associated policies; we assume that they mix between the two with some strictly interior probability.) Note that the system always has three regions: the extreme regions $N$ to the left and $S$ to the right, and an intermediate region that is either $I$ or $M$ depending on the ordering of  $\hat{y}_I$ and $\hat{y}_M$. Numerical examples in Online Appendix \ref{appendix:additional} show that both $\hat{y}_M<\hat{y}_I$ and $\hat{y}_M>\hat{y}_I$ are possible, so the intermediate region can be either of the two. 
\section{Dynamics}\label{subsec:dynamics}
\subsection{The Discrete-Time Stochastic Process}
We begin the analysis of dynamics by describing how the share of true stories evolves in each region. Let  $p^T_R(y),p^F_R(y)$ be  the probabilities that the agent shares a true or false story, respectively, when the current share of true stories is $y$ under the decision rule of region $R\in\{N,I,M,S\}$. These are given by
\begin{equation}\label{probabilities}
   p^T_R(y),\,p^F_R(y)= 
\begin{cases}
y,\, (1-y)\theta\left(1-\delta a (y,I)-(1-\delta)a(y,M)\right),
&R=S \\ 
\frac{y}{2},\,(1-y)\delta\theta \left(1-a(y,I)\right),
&R=I \\

 \frac{y}{2},\,(1-y)(1-\delta) \theta \left(1-a(y,M)\right),
& R=M\\  

0,\,0,
& R=N.
		 \end{cases}
\end{equation}

For example, $p^F_I(y)= (1-y)\delta \theta \left(1-a(y,I)\right)$ because in region $I$ users share a false story if and only if all of the following occur: They drew a false story, the story is very interesting, and they  observed the signal $T'$.

The following Markov processes describe how the system would evolve if users followed the  decision rule of some region $R\in\{N,I,M,S\}$, whether or not it was optimal given the current state:

\begin{equation}\label{lom z_n}
z_{n+1;R} = z_{n;R} +
\begin{cases}
\renewcommand{\arraystretch}{.8}
\begin{pmatrix} 1 + \rho \\ \kappa \end{pmatrix},
& \text{with probability } \quad p^T_R(y_n) \\[14pt]
\renewcommand{\arraystretch}{.8}
\begin{pmatrix} 1 \\ \kappa + \rho \end{pmatrix}, 
& \text{with probability } \quad p^F_R(y_n) \\[14pt]
\renewcommand{\arraystretch}{.8}
\begin{pmatrix} 1 \\ \kappa \end{pmatrix}, 
& \text{w.p} \quad 1 - p^T_R(y_n) - p^F_R(y_n).
\end{cases}
\end{equation}
\medskip

 Appendix \ref{appendix:limit ODE} shows  these processes are \emph{generalized Polya urns} (GPUs), which lets us apply results from \textcite{schreiber2001urn} and \textcite{benaim2004generalized}.

The law of motion for $y_n$ in region $R$ is\footnote{If $z_{n+1}-z_n= \begin{pmatrix}a\\ b\end{pmatrix}$ then $y_{n+1}-y_n = \frac{y_n|z_n|+a}{|z_n|+a+b}-y_n =\frac{(1-y_n)a-y_nb}{|z_n|+a+b}$.} 

\begin{equation}\label{lom y_n}
    y_{n+1} -y_{n}= 
\begin{cases}
\dfrac{(1-y_n)(1+\rho)-y_n\kappa}{|z_n|+1+\kappa+\rho}, 
& \text{with probability}\quad p^T_R(y_n)\bigskip \\ 

\dfrac{(1-y_n)-y_n(\kappa+\rho)}{|z_n|+1+\kappa+\rho},
& \text{with probability}\quad p^F_R(y_n)\bigskip \\

\dfrac{(1-y_n)-y_n\kappa}{|z_n|+1+\kappa} ,
& \text{w.p}\quad 1-p^T_R(y_n)-p^F_R(y_n). \bigskip \\  

		 \end{cases}
\end{equation}

We will use stochastic approximation  to approximate the behavior of the discrete stochastic system $\{y_n\}_{n\ge0}$ by a continuous-time deterministic system. If our system was a single GPU, we could apply results in \textcite{schreiber2001urn} and \textcite{benaim2004generalized} to relate its limit behavior to that of an appropriately chosen \emph{limit differential equation}. Instead, since our system is a concatenation of the GPUs $\{z_{n;R}\}$, we relate its limit behavior to that of a differential inclusion, an equation of the form $\frac{dx}{dt}\in F(x)$ for a set-valued function $F$. We construct this inclusion, which we will refer to as the \emph{limit differential inclusion} or LDI, by pasting together the limit ODEs associated with the GPUs $\{z_{n;R}\}$.  In our model these ODEs are\footnote{See Appendix \ref{appendix:limit ODE} for the derivation of this equation.} 

\begin{equation}\label{general ODE}
   \frac{dy}{dt}= 1+p^T_R(y)\rho - y\big(1+\kappa+\rho(p^T_R(y)+p^F_R(y))\big):= g_R(y).
\end{equation}

For an intuition for the limit ODEs, note that in each region $R$  the expected number of incoming true stories is $1+p^T_R(y)\rho$  and the total expected number of incoming stories is $1+\kappa+\rho\left(p^T_R(y)+p^F_R(y)\right)$. So, 
\begin{equation*}
    \begin{split}
     g_R(y)& = \mathbb{E}_R[\#\text{incoming true stories in period n+1}|y_n=y]\\
     &-y\mathbb{E}_R[\#\text{total incoming stories in period n+1}|y_n=y].  
    \end{split}
\end{equation*}

Thus, according to the limit ODE $\frac{dy}{dt}=g_R(y)$, the share of true stories increases if and only if the ratio of expected incoming true stories to total expected incoming stories  is greater than the current share of true stories.

Our LDI is given by
\begin{equation}\label{inclusion}
 \frac{dy}{dt} \in F(y),  
\end{equation}
where $F(y)=\{g_R(y)\}$ within each region $R$, and at the thresholds, $F$ takes on all values in the interval between the limit ODEs:  If $\hat{y}$ is the threshold between regions $R$ and $R'$, then $
    F(\hat{y}) = [\min\{g_R(\hat{y}),g_{R'}(\hat{y})\},\max\{g_R(\hat{y}),g_{R'}(\hat{y})\}].$

\subsection{Analysis of the Limit Continuous-Time System}
 We say that a point $y^*\in(0,1)$ is a \emph{steady state} for the LDI if $0\in F(y^*)$. We say that $y^*$ is a \emph{stable steady state} for the LDI if it is a steady state and there exists $\epsilon>0$ such that for all $y\in(y^*-\epsilon,y^*+\epsilon)\setminus \{y^*\}$ we have $\sign (x) = \sign(y^*-y)$ for all $x\in F(y)$, and that a steady state is \emph{globally stable} if 
 every solution of the LDI converges to it. We say a steady state is repelling if there exists $\epsilon>0$ such that for all $y\in(y^*-\epsilon,y^*+\epsilon)\setminus\{y^*\}$ we have $\sign (x) = -\sign(y^*-y)$ for all $x\in F(y)$.
 
 We will relate the steady states of the LDI to the behavior of the ODEs in each region. First we note that each of these ODEs has a  globally stable steady state.

 \begin{lemma}\label{lemma unique steady state}
  For all $R\in\{N,I,M,S\}$, the ODE $\frac{dy}{dt}= g_R(y)$ defined over $[0,1]$ has a globally stable steady state $y_R^*\in(0,1)$.
 \end{lemma}

We call the  $y^*_{R}$  \emph{quasi steady states}. The geometry of the phase diagram for the LDI is determined by the relative positions of the thresholds $\hat{y}_I,\hat{y}_M$ and the quasi steady states: The thresholds determine the system's regions, and within each region the flow is towards the corresponding quasi steady state. It is therefore important  to understand the possible orderings of the four quasi steady states. Since the exogenous inflow of true and false stories is constant, any differences between the quasi steady states are due to differences in sharing behavior, so to order them we need to understand how the different decision rules lead to different mixes of true and false stories.

\subsection*{Effects of Users' Sharing Decisions}

 To evaluate the effects of the  decision rules we introduce the following terms: For decision rule $R$ and any point $y\in(0,1)$ we will refer to $p^T_R(y)/y$ and $p^F_R(y)/(1-y)$, respectively, as the \emph{truth-sharing propensity} and \emph{false-sharing propensity} for rule $R$ at $y$. These are the probabilities of sharing a true or false story, conditional on one being drawn. Let $d_R(y):=p^T_R(y)/y-p^F_R(y)/(1-y)$ denote the difference between the truth- and false-sharing propensities. We refer to $d_R(y)$ as users' \emph{discernment} at $y$ when they follow rule $R$. We say that users are \emph{discerning} if $d_R(y)> 0$. The following result shows that discernment is  a key metric for evaluating the impact of users' sharing decisions.  

\begin{lemma}\label{lemma:contribution}
    For any two decision rules $R,W\in\{N,I,M,S\}$,  $y^*_R>y^*_W$ if and only if  $d_R(y^*_R)>d_W(y^*_R)$.
\end{lemma}

 In particular, this lemma shows that if one decision rule consistently implies higher discernment than another, then its corresponding quasi steady state is also higher. 
With decision rules $S$ and $M$ users are always discerning: With rule $M$ users share all stories that are true and mildly interesting, so the truth-sharing propensity is $1/2$, while the false-sharing propensity is below $1-\delta<1/2$; with rule $S$ users share all true stories but only some false stories. Thus, the quasi steady states for these rules are above the quasi steady state for the no-sharing rule. 

Since users are discerning when sharing $M$ stories, sharing both $M$ and $I$ stories always generates a larger net increase in $y$ than sharing $I$ stories alone, so the quasi steady state for rule $S$ is above the quasi steady state for rule $I$.

In contrast to the conclusion for rules $S$ and $M$, users' discernment when only sharing $I$ stories is ambiguous, which is why the relationships between $y^*_S$ and $y^*_M$ and between $y^*_I$ and $y^*_N$ cannot be signed. The ambiguity arises because with rule $I$ the truth-sharing propensity is $1/2$ and the false-sharing propensity at $y$ is $\delta\theta(1-a(y,I))$  and  both $1/2> \delta\theta(1-a(y^*_I,I))$ and the opposite inequality can occur (for different parameters).\footnote{We revisit this issue when discussing comparative statics of $y^*_I$ with respect to $\rho$ in Section \ref{subsec:comparative}. Online Appendix \ref{appendix:additional}  provides numerical examples to show that the inequality can go both ways.} 
Finally, compare decision rules $I$ and $M$. Under decision rule $I$, users consider sharing more false stories than under $M$ because more false stories are of type $I$. Additionally, they have less of an incentive to avoid sharing false stories because the payoff to sharing $I$ stories is greater. Together, these forces imply that users are more discerning with $M$ content than $I$ content.

 Combining the observations above yields the following lemma, where we rule out knife-edge cases of equality between any two of the quasi steady states. 

\begin{lemma}\label{ordering lemma}
    $\min \{y^*_S,y^*_M\}> \max\{y^*_I,y^*_N\}$.
\end{lemma}

The arguments for Lemma \ref{ordering lemma} do not rely on our specific choices of payoffs and signal function, so we expect that this ordering is satisfied for all specifications in which signal precision is increasing in attention and the payoff to sharing an $I$ story is greater than the payoff to sharing an $M$ story.

 Online Appendix \ref{appendix:online} shows that Lemma \ref{ordering lemma} is the only restriction on the ordering of the quasi steady states and thresholds. That is, if the relationship between any two of these variables is not determined by Lemma \ref{ordering lemma} then it can go both ways.   We then explain why this implies that there are 40 possible strict configurations for the five variables that pin down the phase diagram: the two thresholds, and the quasi steady states for the system's three regions, i.e., $y^*_S,y^*_N$ and one of $y^*_I,y^*_M$. As mentioned above, to avoid knife-edge cases we assume that no two quasi steady states are equal. Similarly, we also rule out equality between any two of the quasi steady states and thresholds.

\begin{figure}[h!]
    \centering
    \begin{subfigure}[b]{0.45\textwidth}
        \centering
        (a)
        \begin{tikzpicture}
            \node[] at (-2,1) {\footnotesize REGION N};
            \node[] at (0,1) {\footnotesize REGION I};
            \node[] at (2,1) {\footnotesize REGION S};
            \node[] at (.5,-.3) {\footnotesize $y^*_S$};
            \node[] at (0,-.3) {\footnotesize $y^*_N$};
            \node[] at (-2,-.3) {\footnotesize $y^*_I$};
            \node[] at (1,-.3) {\footnotesize $\hat{y}_M$};
            \node[] at (-1,-.3) {\footnotesize $\hat{y}_I$};
            \node[] at (3,-.3) {\footnotesize 1};
            \node[] at (-3,-.3) {\footnotesize 0};
            \draw [dashed, red] (-1,0) -- (-1,1);
            \draw [dashed, red] (1,0) -- (1,1);
            \draw[very thick, blue, ->] (-3,.5) -- (-1.1,.5);
            \draw[very thick, blue, ->] (.9,.5) -- (-.9,.5);
            \draw[very thick, blue, ->] (3,.5) -- (1.1,.5);
            \draw[draw=violet, fill=violet] (0,0) rectangle ++(0.1,0.3);
            \draw[draw=violet, fill=violet] (.5,0) rectangle ++(0.1,0.3);
            \draw[draw=violet, fill=violet] (-2,0) rectangle ++(0.1,0.3);
            \draw[draw=green, fill=green] (-1,0) rectangle ++(0.1,0.3);
            \draw [very thick] (-3,0) -- (3,0);
        \end{tikzpicture}
    \end{subfigure}
    \hfill
    \begin{subfigure}[b]{0.45\textwidth}
        \centering
        (b)
        \begin{tikzpicture}
            \node[] at (-2,1) {\footnotesize REGION N};
            \node[] at (0,1) {\footnotesize REGION I};
            \node[] at (2,1) {\footnotesize REGION S};
            \node[] at (.5,-.3) {\footnotesize $y^*_S$};
            \node[] at (0,-.3) {\footnotesize $y^*_I$};
            \node[] at (-.5,-.3) {\footnotesize $y^*_N$};
            \node[] at (1,-.3) {\footnotesize $\hat{y}_M$};
            \node[] at (-1,-.3) {\footnotesize $\hat{y}_I$};
            \node[] at (3,-.3) {\footnotesize 1};
            \node[] at (-3,-.3) {\footnotesize 0};
            \draw [dashed, red] (-1,0) -- (-1,1);
            \draw [dashed, red] (1,0) -- (1,1);
            \draw[very thick, blue, ->] (-3,.5) -- (-1.1,.5);
            \draw[very thick, blue, ->] (-.9,.5) -- (-.1,.5);
            \draw[very thick, blue, ->] (.9,.5) -- (.1,.5);
            \draw[very thick, blue, ->] (3,.5) -- (1.1,.5);
            \draw[draw=violet, fill=violet] (-.5,0) rectangle ++(0.1,0.3);
            \draw[draw=violet, fill=violet] (.5,0) rectangle ++(0.1,0.3);
            \draw[draw=green, fill=green] (0,0) rectangle ++(0.1,0.3);
            \draw [very thick] (-3,0) -- (3,0);
        \end{tikzpicture}
    \end{subfigure}
    \vfill
    \begin{subfigure}[b]{0.45\textwidth}
        \centering
        (c)
        \begin{tikzpicture}
            \node[] at (-2,1) {\footnotesize REGION N};
            \node[] at (0,1) {\footnotesize REGION M};
            \node[] at (2,1) {\footnotesize REGION S};
            \node[] at (2,-.3) {\footnotesize $y^*_M$};
            \node[] at (1.5,-.3) {\footnotesize $y^*_S$};
            \node[] at (-2.5,-.3) {\footnotesize $y^*_N$};
            \node[] at (1,-.3) {\footnotesize $\hat{y}_I$};
            \node[] at (-1,-.3) {\footnotesize $\hat{y}_M$};
            \node[] at (3,-.3) {\footnotesize 1};
            \node[] at (-3,-.3) {\footnotesize 0};
            \draw [dashed, red] (-1,0) -- (-1,1);
            \draw [dashed, red] (1,0) -- (1,1);
            \draw[very thick, blue, ->] (-3,.5) -- (-2.5,.5);
            \draw[very thick, blue, ->] (-1.1,.5) -- (-2.4,.5);
            \draw[very thick, blue, ->] (1.1,.5) -- (1.5,.5);
            \draw[very thick, blue, ->] (3,.5) -- (1.6,.5);
            \draw[very thick, blue, ->] (-.9,.5) -- (.9,.5);
            \draw[draw=green, fill=green] (-2.5,0) rectangle ++(0.1,0.3);
            \draw[draw=red, fill=red] (-1,0) rectangle ++(0.1,0.3);
            \draw[draw=green, fill=green] (1.5,0) rectangle ++(0.1,0.3);
            \draw[draw=violet, fill=violet] (2,0) rectangle ++(0.1,0.3);
            \draw [very thick] (-3,0) -- (3,0);
        \end{tikzpicture}
    \end{subfigure}
    \hfill
    \begin{subfigure}[b]{0.45\textwidth}
        \centering
        (d)
        \begin{tikzpicture}
            \node[] at (-2,1) {\footnotesize REGION N};
            \node[] at (0,1) {\footnotesize REGION M};
            \node[] at (2,1) {\footnotesize REGION S};
            \node[] at (2,-.3) {\footnotesize $y^*_S$};
            \node[] at (0,-.3) {\footnotesize $y^*_M$};
            \node[] at (-2,-.3) {\footnotesize $y^*_N$};
            \node[] at (1,-.3) {\footnotesize $\hat{y}_I$};
            \node[] at (-1,-.3) {\footnotesize $\hat{y}_M$};
            \node[] at (3,-.3) {\footnotesize 1};
            \node[] at (-3,-.3) {\footnotesize 0};
            \draw [dashed, red] (-1,0) -- (-1,1);
            \draw [dashed, red] (1,0) -- (1,1);
            \draw[very thick, blue, ->] (-3,.5) -- (-2,.5);
            \draw[very thick, blue, ->] (-1.1,.5) -- (-1.9,.5);
            \draw[very thick, blue, ->] (-.9,.5) -- (-.1,.5);
            \draw[very thick, blue, ->] (.9,.5) -- (.1,.5);
            \draw[very thick, blue, ->] (3,.5) -- (2.1,.5);
            \draw[very thick, blue, ->] (1.1,.5) -- (2,.5);
            \draw[draw=green, fill=green] (0,0) rectangle ++(0.1,0.3);
            \draw[draw=green, fill=green] (2,0) rectangle ++(0.1,0.3);
            \draw[draw=green, fill=green] (-2,0) rectangle ++(0.1,0.3);
            \draw[draw=red, fill=red] (-1,0) rectangle ++(0.1,0.3);
            \draw[draw=red, fill=red] (1,0) rectangle ++(0.1,0.3);
            \draw [very thick] (-3,0) -- (3,0);
        \end{tikzpicture}
    \end{subfigure}
    \caption{Examples of phase diagrams.}
    \label{fig:1}
\end{figure}
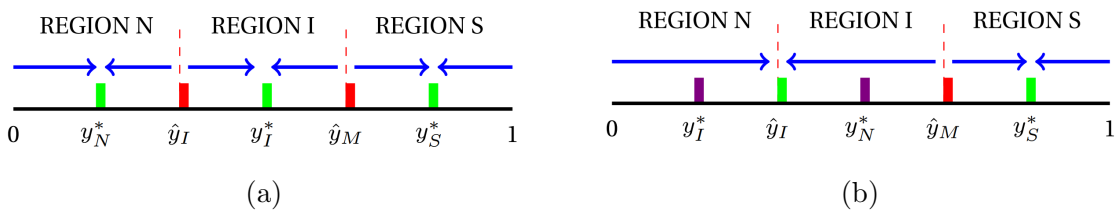

Figure \ref{fig:1} presents four examples of  phase diagrams. The stable steady states of the LDI are  in green, repelling steady states are in red, quasi steady states that are not steady states are in purple, and thresholds are marked by dashed lines. 

Intuitively, one should expect that all quasi steady states that are within their regions are stable steady states for the LDI,  and our next result shows that this is indeed the case.  There can be anywhere from 0 to 3 such steady states, as illustrated in Figure \ref{fig:1}. Since every limit ODE has a unique steady state, the only other candidate steady states for the LDI are the thresholds; these are steady states where the agents are indifferent between two rules and randomize between them.\footnote{Here, as with mixed strategy equilibria in games, the randomizing probabilities are determined by an equilibrium condition.}

For a threshold $\hat y $ to be a stable steady state, the flow above it needs to point down and the flow below it needs to point up. This requires a ``flip'' of quasi steady states: Let $W$ be the region to the left of $\hat y$, and $Z$ the region to the right, a flip is $y^*_Z<\hat y <  y^*_W$. Flips around $\hat{y}_I$ occur when $\hat{y}_I<\hat{y}_M$ and $y^*_I<\hat{y}_I<y^*_N$ (as in phase diagram (a) in Figure \ref{fig:1}), or  $\hat{y}_I>\hat{y}_M$ and $y^*_S<\hat{y}_I<y^*_M$. Online Appendix \ref{appendix:additional} shows that both cases are possible, and Lemma \ref{ordering lemma} implies that flips cannot occur around $\hat{y}_M$. This implies the following characterization of the  set $\mathcal{S}$ of stable steady states.

\begin{theorem}[Stable Steady States]\label{thm S_F}
  Either (a) $\mathcal{S} =\{y^*_R | y^*_R\in \text{region}\, R\}\cup\{\hat{y}_I\}$,  or (b) $\mathcal{S} = \{y^*_R | y^*_R\in \text{region}\, R\}$. Case (a) obtains if and only if  $\hat{y}_I<\hat{y}_M$ and $y^*_I<\hat{y}_I<y^*_N$  or  $\hat{y}_I>\hat{y}_M$ and $y^*_S<\hat{y}_I<y^*_M$.  
\end{theorem}

We use the term \emph{limit points} for values to which $y_n$ converges with positive probability. Since behavior in the no-sharing region ($N$) is deterministic---exactly $1$ true story and $\kappa$ false stories are added every period---if the system starts in region $N$ and $y^*_N\in N$ then $y_n\rightarrow y^*_N=\frac{1}{1+\kappa}$ deterministically. Otherwise, any stable steady state is a limit point.

\begin{theorem}[Limit Points]\label{thm limit behavior}
$y_n$ converges almost surely to a point in $\mathcal{S}$. If $y^*_N\in N$ and $y_0\in N$ then $y_n$ converges to $y^*_N$.  Otherwise, for all $y^*\in\mathcal{S}$ there is positive probability that $y_n$ converges to $y^*$. 
\end{theorem}

The proof of Theorem \ref{thm limit behavior} has three parts. First, Theorem \ref{thm limit chain transitive} in Appendix \ref{appendix:stochastic approximation} shows that $y_n$ almost surely converges to a steady state of the LDI. Second, Lemma \ref{claim stable} in Appendix \ref{appendix:proofs} shows that every stable steady state has positive probability of being the limit point. Finally, Lemma \ref{claim repelling} in Appendix \ref{appendix:proofs} shows that the system almost surely does not converge to a repelling steady state. This completes the proof, because our simplifying assumption that no two variables in $\{\hat{y}_I,\hat{y}_M,y^*_N,y^*_I,y^*_S\}$ are equal means that any steady state is either stable or repelling.\footnote{For a threshold $\hat{y}$ to be a steady state that is neither stable or repelling, the flow must have the same sign on both sides of $\hat{y}$. This is only possible at a  threshold that is also a quasi steady state, which  we have ruled out.}

\subsubsection*{Detailed Proof Summary}

   Theorem \ref{thm limit chain transitive} in  Appendix \ref{appendix:stochastic approximation} relates the limit behavior of concatenations of GPUs to the asymptotic behavior of the differential inclusions that  concatenate the corresponding ODEs.
     Applied to our system, the theorem  implies that the limit set of $y_n$, $L(y_n) = \bigcap_{m> 0}\overline{\{y_n : n>m\}}$, is almost surely a steady state of the LDI.\footnote{Here overline  denotes the closure. The proof of Theorem \ref{thm limit chain transitive} shows that a continuous-time interpolation of $y_n$ is a perturbed solution to the LDI in the sense of BHS, and then  applies a result in BHS that characterizes limits of perturbed solutions. (See appendix \ref{appendix:stochastic approximation} for definitions of these terms.) }

To prove Lemma \ref{claim stable}, that  there is positive probability of convergence to every stable steady state, we  first show that $y_n$ has positive probability of converging to any $y^*_R$ conditional on starting from states $z_m$ with $|z_m|$ sufficiently large and $y_m$ sufficiently close to $y^*_R$. This claim is true for a counterfactual process that follows the decision rule of region $R$ everywhere, because that process converges almost surely to $y^*_R$. This implies that the claim is also true for $y_n$, because: i) when $y_n$ is in region $R$ it follows the same law of motion as the counterfactual process, and ii) as we show, starting from a state $z_m$ with $|z_m|$ sufficiently large and $y_m$ sufficiently close to $y^*_R$ the counterfactual process (and therefore also $y_n$) has positive probability of never leaving region $R$. We complete the proof for the quasi steady states  by showing that the system has positive probability of arriving at a state $z_m$  from which convergence occurs with positive probability. The proof for the case where the stable steady state is $\hat{y}_I$ is similar but uses a different counterfactual process. 

Finally, the proof of Lemma \ref{claim repelling}, that $y_n$ almost surely does not converge to a repelling steady state,  uses Theorem \ref{thm nonconvergence} in Appendix \ref{appendix:stochastic approximation}, which shows that a sufficient condition for nonconvergence to a repelling steady state is that there is a positive uniform lower bound on the noise in the stochastic process. Intuitively, noise jiggles $y_n$ away from the steady state, and because the steady state is repelling, the drift of the process will tend to move it further away.

\subsection*{Discussion}

 Our simplified representation of platform dynamics allows for rich limit behavior. Our finding that the limit share of true stories is random, though not mathematically surprising within the context of generalized urns, has notable implications for the evolution of platform composition. It implies that starting from the same initial platform composition and parameters, the system can end up at very different limits in terms of both the share of true stories and users' limit actions. For instance, in some cases the system has positive probability of converging to any  of three limits: One in which the share of true stories is low and users do not share at all (since the probability of sharing a false story is high), one in which the share of true stories is intermediate and users share either only mildly interesting stories or only very interesting ones, and one in which the share of true stories is high and users share both very interesting and mildly interesting stories. This path-dependence suggests that the long-run outcome can be influenced by shocks to the platform, such as a sudden influx of false stories. These shocks will be more likely to change limit behavior if they occur early, when the overall number of stories is small. 
 
In this context, our model can be interpreted as tracking the entire universe of stories on a platform or channel. Alternatively, it could apply to stories on a single issue (e.g., COVID-19, a presidential election, etc.), if users take into account only the prevalence of false stories on that issue in their sharing decisions. In both cases, our results imply that the nature of early discourse on a channel or issue could have long-lasting implications. A suggestive illustration appears in \textcite{zhang2021understanding}, which considers two Reddit communities dedicated to COVID-19 that began with similar user bases and content. Following an early platform-level shock in which one was designated the official community and subjected to stricter moderation, the two rapidly diverged, with the official community attracting science enthusiasts and the other becoming a hub for conspiracy theories.   

\section{Comparative Statics}\label{subsec:comparative}

 This section explains how 
 the limit points change with respect to the credibility parameter $\theta$, the reach parameter $\rho$, and the false story production rate $\kappa$---the parameters for which the comparative statics are the most interesting.
  Online Appendix \ref{appendix:comparative} gives the results and proofs that underlie  this section as well as comparative statics  with respect to the other parameters.\footnote{Among other results, we show that the limit share of true stories can  either increase or decrease in the probability that false stories are true $(\delta)$ and the cost of attention $(\beta)$. The counterintuitive result for $\beta$ arises when the limit point is $\hat{y}_I$.}

 \bigskip 
 
  \noindent \textbf{The role of $\theta$}
  \begin{table}[h]
  \centering
  \caption{\textbf{Comparative Statics for \(\theta\)}}
  \label{table:theta}
  \setlength{\tabcolsep}{12pt} 
  There are  switchpoints $\theta_M,\theta_I,\theta_S\in(0,1]$ such that:
  \begin{tabular}{ll@{}}
  \toprule
  \(y^*_M\) & Decreasing for \(\theta < \theta_M\) and increasing for \(\theta > \theta_M\). \\
  \(y^*_S\) & Decreasing for \(\theta < \theta_S\) and increasing for \(\theta > \theta_S\). \\
  \(y^*_I\) & Decreasing for \(\theta < \theta_I\) and increasing for \(\theta > \theta_I\). \\
  \(\hat{y}_I\) & Increasing. \\
  \bottomrule
  \end{tabular}
  \end{table}

Numerical examples reported in Online Appendix \ref{appendix:additional} show that each of the switchpoints in Table \ref{table:theta} can be interior. When they are, the associated quasi steady state is decreasing in $\theta$ up to a point and then increasing.\footnote{The candidate limit point $\hat{y}_I$ behaves differently:  Users' payoffs from sharing are decreasing in the credibility of false stories, so $\hat{y}_I$ needs to increase to maintain indifference.} The intuition is that when the credibility of false stories increases it is harder to identify false stories, but users are aware of this and pay more attention: Both $a(y,I)$ and $a(y,M)$ are increasing in $\theta$. This leads to two opposing forces on the limit share of true stories, and our model predicts that either one can prevail. That is, for sufficiently large values of $\theta$ the increase in attention more than compensates for the increase in credibility.

To illustrate, fix all parameters except the credibility $\theta$ at levels that ensure that $y^*_S$ is the unique limit point of the system for any value of $\theta$ and that the switchpoint $\theta_S$ is interior in $(0,1)$ (see Online Appendix \ref{appendix:additional} for an example of such a configuration). Suppose there is a single false news producer who pays a cost that increases in $\theta$, and obtains a payoff that is increasing in the limit share of false news. Then the producer would not choose a credibility level above the switchpoint $\theta_S$, even if credibility is free. Thus, the overall prevalence of false stories may be higher in a world where all false stories are of limited credibility than in one where they are indistinguishable from the truth. Consequently, false news producers may choose intermediate credibility levels to benefit from the ``lulling effect" associated with lower credibility. This could provide a partial rationale for propaganda strategies such as the ``firehose of falsehood," a term coined by \textcite{paul2016russian} to refer to contemporary Russian propaganda, characterized by a ``shameless willingness to disseminate partial truths or outright fictions."  

Another interpretation of $\theta$ is that the social media platform implements a fact-checking scheme that never mislabels true stories as false, with $\theta$ the probability that a false story is \emph{not} flagged as false. Under this interpretation, the comparative statics of the quasi steady states with respect to $\theta$ imply that if flagging rates are low ($\theta$ is high), marginally improving them may have unintended consequences. Again, the intuition relates to a counterbalancing force driven by attention choices. When more stories are flagged, users pay less attention. This means they are more likely to share stories that have not been flagged, which can lead to an overall increase in the limit share of false stories.

Under this interpretation, the comparative statics for the quasi steady states $\{y^*_S,y^*_I,y^*_M\} $ are closely related to what \textcite{pennycook2020implied} calls the ``implied truth effect'': the idea that, in the presence of flagging, content that is not flagged as false is considered validated. \textcite{pennycook2020implied} shows how this effect can arise through Bayesian updating, and demonstrates it in experiments. Instead of considering the effect of flagging on independent individual decisions, we consider its effect on limit platform composition. In our setting, the effect is mediated by attention, with users paying less attention to stories that have not been flagged. Our results show that the implied truth effect may outweigh the direct beneficial effect of flagging so that improving flagging rates may increase the prevalence of  false stories on the platform. Relatedly, \textcite{acemoglu2023online} finds that the interaction of the implied truth effect with a platform's engagement-maximizing incentives may lead a regulator to censor less misinformation than is technologically feasible. In that  model, the regulator always prefers some censorship to none. In ours, there are cases  where some flagging instead of none increases the spread of false news.\footnote{No flagging corresponds to $\theta=1$, and we saw above that this can lead to a higher share of true stories than some values of $\theta<1$.}

Finally, the comparative statics with respect to $\hat{y}_I$ imply that the limit share of true stories may be everywhere decreasing in the flagging rate, through the constraint that users are indifferent, a mechanism distinct from the implied truth effect.

\bigskip 

\noindent\textbf{The role of $\rho$}
\begin{table}[h]
\centering
\caption{\textbf{Comparative Statics for $\rho$}}
\label{table:rho}
\setlength{\tabcolsep}{10pt}
\begin{tabular}{ll@{}}
\toprule
\(y^*_M\) & Increasing. Converges to $1$ as $\rho\rightarrow\infty$. \\
\(y^*_S\) & Increasing in $\rho$. Converges to $1$ as $\rho\rightarrow\infty$. \\
\(y^*_I\) & Increasing if $d_I(\frac{1}{1+\kappa})>0$, decreasing if the inequality is reversed. \\
&Converges to either $0,1$ or to an interior value as  $\rho\rightarrow\infty$ \\

\(\hat{y}_I\) & Constant. \\
\bottomrule
\end{tabular}
\end{table}

The reach parameter has no effect on the location of the threshold $\hat{y}_{I}$ because it is not an argument in users' payoffs. To understand the effects of $\rho$ on the quasi steady states, recall that $d_R(y)$ is the difference between users' truth- and false-sharing propensities when they follow decision rule $R$ at point $y$. Under decision rules  $S$ and $M$, it is always the case that $d_R(y)>0$: When users follow either rule, the net increase in the share of true stories is larger than when they do not share, which explains why the corresponding quasi steady states increase when the number of copies of each story shared by users increases. However, $d_I(y)$ can be either positive or negative, which is why the effect of reach on the quasi steady state for this rule is ambiguous. We explain below why it suffices to consider users'  discernment with rule $I$ at the no-sharing quasi steady state $y^*_N=\frac{1}{1+\kappa}$ to sign the effect of $\rho$ on $y^*_I$.

As the reach parameter increases to infinity, the quasi steady states $y^*_S,y^*_M$ converge to $1$ while $y^*_I$ can converge to $0$, to $1$, or to some interior value $\bar y\in (0,1)$, depending on the parameters. The limit is interior if and only if there exists $\bar y\in(0,1)$ at which $d_I(\bar y)=0$. We find that the effect of attention on $y^*_I$ always works in the opposite direction as the effect of reach, so that the change due to $\rho$ is counterbalanced and may eventually settle down.

To illustrate the effect of $\rho$ on $y^*_I$, suppose that $d_{I}(y)$ is positive at $y^*_N$, and equal to zero at some interior $\bar y > y^*_N$. When $\rho=0$, all quasi steady states equal $y^*_N=1/(1+\kappa)$. Since users are discerning with rule $I$ at $y^*_N$, a marginal increase in $\rho$ increases $y^*_I$. 
Since attention is decreasing in $y$, so is discernment. Thus, as $y^*_I$ increases in $\rho$, users' discernment decreases. But as long as discernment is positive, $y^*_I$ continues to increase in $\rho$. It cannot increase beyond $\bar y$ since at all points above $\bar y$ discernment is negative. Hence, $y^*_I$ monotonically converges to $\bar y$. This also explains why $1/(1+\kappa)$ appears in the condition in Table \ref{table:rho}: If $y^*_I(\rho)$ is increasing (decreasing) at $\rho=0$, i.e., when $y^*_I=1/(1+\kappa)$, then it is increasing (decreasing) for all $\rho$.  

 Our analysis highlights that higher reach does not inherently lead to a greater spread of false news. Instead, the impact depends on how much users prioritize sharing very interesting stories and the prevalence of false stories within the system. 

 \bigskip 

 \noindent\textbf{The role of $\kappa$}
 
In contrast to  the parameters discussed above, the local effects of  changes in $\kappa$ are obvious: All quasi steady states are strictly decreasing in $\kappa$, and both thresholds are constant in $\kappa$, since it is not an argument in users' payoffs (see Online Appendix \ref{appendix:comparative}). The effect of $\kappa$ on which steady states are stable is more interesting.

\begin{theorem}\label{thm: comparative kappa}
    Fix values for the other parameters and let $\mathcal{S}_\kappa$ be the set of stable steady states as a function of $\kappa$. There exist $0<\kappa_1<\kappa_2<\infty$ such that 
    \begin{enumerate}
  \item $\mathcal{S}_\kappa$ is equal to $\{y^*_S\}\cup A$ for all $\kappa<\kappa_1$, where $A=\emptyset$ or $\{y^*_I\}$ or $\{\hat {y}_I\}$, depending on the other parameters. 
\item $\mathcal{S}_\kappa=\{y^*_N\}$ for all $\kappa>\kappa_2$. 
    \end{enumerate}
\end{theorem}

As $\kappa\rightarrow \infty$, the exogenous inflow of false stories dominates any inflow due to user sharing, pushing all quasi steady states to $0$. Eventually, false stories become so prevalent that users stop sharing altogether, so the unique limit point is $y^*_N$.  
At the other extreme, the proof of Theorem \ref{thm: comparative kappa} shows that when users are discerning with decision rule $I$, or when the intermediate region is $M$, then $\mathcal{S}_\kappa=\{y^*_S\}$ for all sufficiently small $\kappa$. In these cases, increasing the production rate of false stories from low to high shifts limit behavior from sharing both very interesting and mildly interesting stories to not sharing at all. We saw above that users are discerning when they share stories of both interest levels, so here the exogenous decrease in the share of incoming stories that are true is amplified by user behavior.\footnote{Relatedly, some changes in $\kappa$ lead to discontinuous jumps in the distribution of $\lim_{n\rightarrow\infty}y_n$. This happens when a quasi steady state crosses a threshold so that it (or the threshold) is no longer a limit point.} However, if users are not discerning when they follow decision rule $I$, then either $y^*_I$ or $\hat{y}_I$ may be an additional limit point for small $\kappa$, in which case a significant share 
of false stories can persist in the limit even as the exogenous inflow of false 
stories vanishes. See Online Appendix \ref{appendix:additional} for examples. 
 
\subsection*{Welfare}

Recall that $V(y,\ell) := U(a(y,\ell),y,\ell)$ denotes a users' expected value from sharing a story with interest level $\ell$ when there is a share $y$ of true stories on the platform. Define expected welfare as 
\[
    W(y):= \mathbb{P}(I|y)\max\{V(y,I),0\}+ \mathbb{P}(M|y)\max\{V(y,M),0\},
\]
where $\mathbb{P}(I|y),\mathbb{P}(M|y)$ are the probabilities of drawing a very interesting and mildly interesting story, respectively, when the fraction of true stories is $y$.\footnote{So $\mathbb{P}(I|y)=\frac{1}{2}y+\delta(1-y)$, and $\mathbb{P}(M|y)=\frac{1}{2}y+(1-\delta)(1-y)$. } In the no-sharing region welfare is constant at $0$, in the other regions, welfare is strictly increasing in the fraction of true stories.

\begin{lemma}\label{lemma welfare}
 $W(y)$ is strictly increasing for $y\notin N$.\footnote{This result does not follow immediately from Lemma \ref{lemma thresholds}, because as $y$ increases the share of very interesting stories decreases.} 
\end{lemma}

Now, we consider the welfare impact of changing the  reach parameter $\rho$. Because reach  does not directly affects users' payoffs, and welfare is increasing in the fraction of true stories, in any given limit point the effect of reach on welfare works in the same direction as its effect on the share of true stories. However, increasing $\rho$ will change the probability of converging to each limit point, so we can only sign the effect on welfare for configurations in which either (i) $y^*_I$ is increasing in $\rho$,  (ii) $y^*_I$ is the only limit point, or (iii)  $y^*_I$ is not a limit point.  Analyzing the welfare effects of the credibility parameter is harder since this parameter has a direct negative effect on welfare as well as 
an indirect effect through the fraction of true stories. Any conclusions here would rely heavily on our functional form assumptions, which we only view as an approximation.

\section{Conclusion}\label{sec:conclusion}

This paper analyzes a model of the sharing of stories on a  social media platform when users' attention levels are endogenous and depend on the mix of true and false stories. The share of true stories converges almost surely, but the realized limit point is stochastic, and  different possible limits have very different  user sharing behavior. This randomness of the limit implies that the type of stories users happened to be exposed to in the early days of the platform and their subsequent sharing decisions can have long-term implications.

 The limit share of true stories may be either increasing or decreasing in each of the following parameters: the cost of attention, the credibility of false stories, the probability that false stories are very interesting, and the reach of shared stories.  Although  endogenous attention creates a counterbalancing force to changes in the credibility/flagging of false stories, it can intensify the effect of producing more false stories. This suggests that interventions that target producers of false news might be more efficient than attempts to stop the spread of false news already on the platform. 

Our model captures many important features in a tractable framework, and parts with most of the literature by tracking the evolution of the entire platform rather than the spread of a single story. Its key simplifying feature is that it has a one-dimensional state space. We maintain this feature while considering two-dimensional story characteristics by assuming that only a story's veracity is fixed while each user's interest is drawn independently each period.
It  would be  straightforward to analyze variations that preserve this structure. For instance, \textcite{allcott2017social} shows that   education, age, and total media consumption are strongly associated with discernment between true and false content. This user heterogeneity can be incorporated into our model by having the user's type drawn randomly every period. \textcite{allcott2017social} also finds that in the run-up to the 2016 election, both Democrats and Republicans were more likely to believe ideologically aligned articles than nonaligned ones. Such partisan considerations can be incorporated by having both the user's and story's partisanship drawn every period. 

Other important features of social media behavior could in principle be handled with similar techniques but  a larger state space. Models where some stories are always more interesting or where users care about additional (fixed) story characteristics could be analyzed as a concatenation of urn models with more colors of balls. Extending our stochastic approximation arguments to these settings is straightforward, but analyzing the associated deterministic continuous-time dynamics is more complex as they would be described by differential inclusions in two or more dimensions. Yet other features do not fall within the urn-based formulation described here. For example, our model does not  track the number of times an individual story has been shared, so it does not capture the ``illusory truth'' effect described in \textcite{pennycook2018prior}, where users perceive stories they have seen many times as more likely to be true.

\begin{appendix}
\titleformat{\section}
  {\normalfont\Large\bfseries}{Appendix \thesection:}{1em}{}

\section{Proofs}\label{appendix:proofs}
\begin{proof}[Proof of Lemma \ref{lemma U a V}]
When  $v=T$,  then $s=T'$ with probability $1$ and $\ell=I$ with probability $\frac{1}{2}$. When $v=F$, then $\ell=I$ with probability $\delta$.  Thus,
\begin{equation*}
    \mathbb{P}_{a,y}(T',T|I)=\frac{\mathbb{P}_{a,y}(T',T,I)}{\mathbb{P}_{a,y}(I)}=\frac{\frac{y}{2}}{\frac{y}{2}+(1-y)\delta}=\frac{y}{y+2(1-y)\delta}.  
\end{equation*}

Similarly, $\mathbb{P}_{a,y}(T',T|M)= \dfrac{y}{y+2(1-y)(1-\delta)}$, $ 
      \mathbb{P}_{a,y}(T',F|I) = \dfrac{2(1-y)\delta \theta(1-a)}{y+2(1-y)\delta}$, and $
    \mathbb{P}_{a,y}(T',F|M) = \dfrac{2(1-y)(1-\delta) \theta(1-a)}{y+2(1-y)(1-\delta)}.$ We can rewrite \eqref{definition of U} as
\begin{equation*}
        U(a,y,M) = \mathbb{P}_{a,y}(T',T|M)u(T,M)+\mathbb{P}_{a,y}(T',F|M)u(F,M)-\beta a^2.
\end{equation*}
Since $u(T,M)= 1$, and $u(F,M)=1-\mu$, we have
\[
 U(a,y,M) = \frac{y -2(\mu-1)(1-y)(1-\delta) \theta}{y+2(1-y)(1-\delta)}+\frac{2(\mu-1)(1-y)(1-\delta) \theta}{y+2(1-y)(1-\delta)}a-\beta a^2.
\]
Similarly, $u(T,I)= 1+\lambda$ and $u(F,I)=1+\lambda-\mu$ implies that
\[
U(a,y,I) = \frac{(1+\lambda)y-2(\mu-1-\lambda)(1-y)\delta \theta}{y+2(1-y)\delta} +\frac{2(\mu-1-\lambda)(1-y)\delta \theta}{y+2(1-y)\delta}a-\beta a^2.
\]

The functions $U(a,y,I)$ and $U(a,y,M)$ are strictly concave in $a$, and they are maximized at 
$a(y,I),a(y,M)$ respectively as defined in Lemma \ref{lemma U a V}. Finally, using Assumptions \ref{condition on lambda} and \ref{sufficient for a below 1} it is straightforward to verify that $a(y,I),a(y,M)\in[0,1]$.\end{proof}
 
\begin{proof}[Proof of Lemma \ref{lemma thresholds}]
Plugging the optimal attention levels from \eqref{attention levels} back into $U(a,y,M),U(a,y,I)$ respectively we get, 
\begin{equation}\label{value functions}
\begin{split}
    V(y,M) &= \frac{y -2(\mu-1)(1-y)(1-\delta) \theta}{y+2(1-y)(1-\delta)}+\frac{1}{\beta}\left(\frac{(\mu-1)(1-y)(1-\delta) \theta}{y+2(1-y)(1-\delta)}\right)^2 ,\\
    V(y,I) &= \frac{(1+\lambda)y-2(\mu-1-\lambda)(1-y)\delta \theta}{y+2(1-y)\delta} +\frac{1}{\beta}\left(\frac{(\mu-1-\lambda)(1-y)\delta \theta}{y+2(1-y)\delta}\right)^2.
\end{split}
\end{equation}
To prove that these value functions are strictly increasing in $y$, it suffices to show that $U(a,y,M),U(a,y,I)$ are strictly increasing in $y$ for all $a$, as then for $y_2>y_1$ we have $V(y_1) = U(a(y_1),y_1)<U(a(y_1),y_2)\le U(a(y_2),y_2) = V(y_2)$. These functions are increasing because
\begin{equation*}
\begin{split}
     \frac{\partial U(a,y,M)}{\partial y} & =\frac{2 (1-\delta) (1+(1-a) \theta  (\mu -1))}{\left(y+2(1-y)(1-\delta)\right)^2} >0 \\
    \frac{\partial U(a,y,I)}{\partial y}& = \frac{2 \delta  ((a-1) \theta  (\lambda -\mu +1)+\lambda +1)}{(2 \delta -2 \delta  y+y)^2} = \frac{2\delta\left(1+\lambda+(1-a)\theta(\mu-1-\lambda)\right)}{\left(y+2(1-y)\delta\right)^2}  >0,
\end{split}
\end{equation*}
where the inequalities follows from Assumption \ref{condition on lambda}.
To see that both $\hat{y}_I,\hat{y}_M$ are interior, note that $V(1,M) = 1>0, V(1,I) = 1+\lambda>0$, and, by Assumptions \ref{condition on lambda} and \ref{sufficient for a below 1}, 
\[
\begin{split}
     V(0,M) &= (\mu-1)\theta\left( \frac{(\mu-1)\theta}{4\beta}-1\right)<0,\\
       V(0,I) &= (\mu-1-\lambda)\theta\left(\frac{(\mu-1-\lambda)\theta}{4\beta}-1\right)<0. 
\end{split}
\]
\end{proof}

\begin{proof}[Proof of Lemma \ref{lemma unique steady state}]
First, note that by the definition of $g_R(y)$ in \eqref{general ODE}, for all $R\in\{N,I,M,S\}$ we have $g_R(0)=1$ and $g_R(1)=-\kappa$. This follows from $g_R(0)=1+p^T_R(0)\rho$ and $p^T_R(0)=0$ for all $R$, and $g_R(1)= -\kappa -p^F_R(1)\rho$ and $p^F_R(1) = 0$ for all $R$. For $R=N$, the ODE takes the simple form $g_N(y) = 1-(1+\kappa)y$ and the conclusion follows immediately with $y^*_N = \frac{1}{1+\kappa}$. For the other regions, it suffices to prove that $g_R'''(y)>0$ for all $y\in[0,1]$. Indeed, for $g_R(y)$ to have more than one root in $[0,1]$ it must have a local minimum that is greater than the first root, followed by a local maximum (between the second root and $y=1$). So, there need to be $0<w<z<1$ such that $g_R''(w)\ge 0$ while $g_R''(z)\le 0$ which cannot be the case if $g_R'''(y)>0$ for all $y\in[0,1]$. The derivatives are 
\[
\begin{split}
g'''_S(y)&= \frac{12 \theta ^2 \rho }{\beta}\left(\frac{\delta ^3  (\mu -1-\lambda)}{ (y+2 (1-y)\delta)^4}+ \frac{(1-\delta)^3(\mu -1)}{(y+2(1-y)(1-\delta))^4} \right), \\
g_I'''(y) &=  \frac{12\rho \delta ^3 \theta^2  (\mu -1-\lambda)}{\beta  (y+2 (1-y)\delta)^4},\\
g_M'''(y) &=   \frac{12 \rho(1-\delta)^3 \theta^2  (\mu -1)}{\beta  (y+2(1-y)(1-\delta))^4}.   
\end{split}
\]

By Assumption \ref{condition on lambda}, all are strictly positive for $y\in[0,1]$. Stability follows from the existence of a unique root together with $g_R(0)=1>0,g_R(1)=-\kappa<0$ for all  $R$. 
\end{proof}

\begin{proof}[Proof of Lemma \ref{lemma:contribution}] Fix two decision rules $R,W\in \{N,I,M,S\}$. By Lemma \ref{lemma unique steady state}, $y^*_R>y^*_W$ if and only if $g_R(y^*_R)=0>g_W(y^*_R)$. By \eqref{general ODE}, for all $y\in[0,1]$, 
\begin{equation*}
    g_R(y)-g_W(y) = \rho\left[(1-y)\left(p^T_R(y)-p^T_W(y)\right)-y\left(p^F_R(y)-p^F_W(y)\right) \right].
\end{equation*}
So $g_R(y^*_R)>g_W(y^*_R)$ if and only if $(1 -y^*_R)\left(p^T_R(y^*_R)-p^T_W(y^*_R)\right)>y^*_R\left(p^F_R(y^*_R)-p^F_W(y^*_R)\right)$. By Lemma \ref{lemma unique steady state}, all quasi steady states are interior in $[0,1]$, so we can divide by $(1-y^*_R)y^*_R$ and rearrange to get 
$y^*_R>y^*_W\iff g_R(y^*_R)>g_W(y^*_R)\iff d_R(y^*_R)>d_W(y^*_R).$
 \end{proof}

\begin{proof}[Proof of Lemma \ref{ordering lemma}]
By Lemma \ref{lemma unique steady state}, to prove $y^*_R>y^*_W$ for $R,W\in \{N,I,M,S\}$, it suffices to prove that for all $y\in(0,1)$: $d_R(y)>d_W(y)$. Fix $y\in(0,1)$. We will show that $\min\{d_S(y),d_M(y)\}>\max\{d_I(y),d_N(y)\}$.
By  \eqref{probabilities}, 
$d_S(y) = 1-\theta\left(1-\delta a (y,I)-(1-\delta)a(y,M)\right)$,
 $d_M(y)  = \frac{1}{2} - (1-\delta) \theta \left(1-a(y,M)\right), 
d_I(y) = \frac{1}{2} -\delta\theta \left(1-a(y,I)\right)$ and $d_N(y) = 0.$

So $\min\{d_S(y),d_M(y)\}>d_N(y)$ follows from the assumptions $\delta>\frac{1}{2},\theta<1$ and since by Lemma \ref{lemma U a V} both attention levels are nonnegative. Similarly, $d_S(y)>d_I(y)$ if and only if $\frac{1}{2}>\theta(1-\delta)(1-a(y,M))$, which always holds. Finally, to prove that $d_M(y)>d_I(y)$ it suffices to prove $
(1-\delta)\left(1-a(y,M)\right) < \delta \left(1-a(y,I)\right).$
Let $q(\delta) =(1-\delta)\left(1-a(y,M)\right); r(\delta) =\delta \left(1-a(y,I)\right)$. We will prove $q(\delta)<r(\delta)$ for all $\delta\in[\frac{1}{2},1)$ by showing that $q(\frac{1}{2})<r(\frac{1}{2})$ and $q(\delta)$ is decreasing in $\delta$ while $r(\delta)$ is increasing in $\delta$. First, 
\begin{equation*}
    r(1/2)  =\frac{1}{4}\left(2-\frac{\theta  (1-y) (\mu-1-\lambda)}{ \beta }\right) >  \frac{1}{4} \left(2-\frac{\theta  (1-y)(\mu -1) }{\beta }\right)= q(1/2).
\end{equation*}
Next, 
\[
\begin{split}
    \frac{\partial q(\delta)}{\partial \delta}& = \frac{2 (1-\delta) \theta  (\mu -1) (1-y) (1-\delta  (1-y))}{\beta (y+2(1-y)(1-\delta))^2}-1\\
    \frac{\partial r(\delta)}{\partial \delta} &= 1-\frac{2 \delta  \theta  (1-y) (\mu -1-\lambda) ( \delta+y(1-\delta) )}{\beta(y+2 (1-y)\delta)^2}
    \end{split}
\]
Assumption \ref{sufficient for a below 1} and  $\lambda>0$ imply that $ \theta(\mu-1-\lambda )<\theta(\mu-1) < 2\beta$; simple algebra then shows that $\frac{\partial q(\delta)}{\partial \delta}<0$ and $\frac{\partial r(\delta)}{\partial \delta}>0$.\end{proof}

\begin{proof}[Proof of Theorem \ref{thm S_F}]
That $\{y^*_R | y^*_R\in \text{region}\, R\}\subset \mathcal{S}$ follows immediately from the definitions of these sets and of $F$. Since each limit ODE has a unique steady state, the only other possible  members of $\mathcal{S}$ are the thresholds between the regions, so $\mathcal{S}\subset \{y^*_R | y^*_R\in \text{region}\, R\}\bigcup\{\hat{y}_I,\hat{y}_M\}$. A threshold $\hat y $ is a stable steady state if for all $y\in(\hat{y}-\epsilon,\hat{y}+\epsilon)$ we have $\sign (x) = \sign(\hat{y}-y)$ for all $x\in F(y)$. This holds only if there is a ``flip'' of quasi steady states: 
 Let $W$ be the region to the left of $\hat y$, and $Z$ the region to the right, a flip is: $y^*_Z<\hat y <  y^*_W$. Flips around $\hat{y}_I$ occur if and only if one the following holds: $\hat{y}_I<\hat{y}_M$ and $y^*_I<\hat{y}_I<y^*_N$; or  $\hat{y}_I>\hat{y}_M$ and $y^*_S<\hat{y}_I<y^*_M$. In Appendix \ref{appendix:online} we show that both are possible. We now show that flips cannot occur around $\hat{y}_M$ so $\hat{y}_M\notin\mathcal{S}$. There are two possible cases:
\begin{enumerate}
    \item  $\hat{y}_I<\hat{y}_M$, so the region to the right of $\hat{y}_M$ is $S$ and the region to the left is $I$. 
    \item $\hat{y}_I>\hat{y}_M$, so the region to the right of $\hat{y}_M$ is $M$ and the region to the left is $N$.
\end{enumerate}
In Case 1 a flip cannot occur because by Lemma \ref{ordering lemma}, $y^*_S>y^*_I$. In Case 2 a flip cannot occur because by Lemma \ref{ordering lemma}, $y^*_M>y^*_N$. 
\end{proof}
 
\begin{proof}[Proof of Theorem \ref{thm limit behavior}]
   When $y^*_N\in N$ and $y_0\in N$, the system follows the law of motion $z_{n+1} = z_n+\begin{pmatrix}1\\ \kappa\end{pmatrix}$, so  it never leaves the region $N$ and converges deterministically to $y^*_N=\frac{1}{1+\kappa}$.  We henceforth assume that $y^*_N \notin N$ and/or  $y_0\notin N$. 
By Theorem \ref{thm limit chain transitive} in Appendix \ref{appendix:stochastic approximation}, the limit set of $y_n$ is almost surely internally chain transitive for the LDI  \eqref{inclusion}. Since the LDI is a one-dimensional autonomous inclusion with isolated steady states, its internally chain transitive sets are simply its steady states, so $y_n$ converges almost surely to a steady state of the LDI.  By Lemma \ref{claim stable} below, when $y^*_N \notin N$ and/or $y_0\notin N$ there is positive probability of convergence to any stable steady state, and by Lemma \ref{claim repelling} there is zero probability of convergence to any repelling steady state.
\end{proof}
 
 Lemma \ref{claim stable} and Lemma  \ref{claim repelling} below are used to prove Theorem \ref{thm limit behavior}, and Lemmas \ref{lemma connected} and \ref{lemma trapping} are  used to prove Lemma \ref{claim stable}.

 \begin{lemma}\label{lemma connected}
      Let $\epsilon>0$ and $y\notin N $ such that  $y\in(\frac{1}{1+\kappa+\rho},\frac{1+\rho}{1+\kappa+\rho})$. Starting from any state $z_n$ with $y_n\notin N$, the system has positive probability of arriving at some $y_m\in B_\epsilon(y)$.
 \end{lemma}
 \begin{proof}
Since at most $1+\kappa+\rho$ stories are added each period, $|y_{n+1}-y_n|\le(1+\kappa+\rho)/|z_n|$, so there exists $n_\epsilon\in\mathbb{N}$ such that $|y_{n+1}-y_n|<\epsilon$ whenever $|z_n|>n_\epsilon$. Since $|z_n|\to\infty$, we may assume $|z_n|>n_\epsilon$ (any finite sequence of consecutive true-story shares occurs with positive probability and keeps $y_n$ outside $N$ while $|z_n|$ grows: each share moves $y_n$ toward $\frac{1+\rho}{1+\kappa+\rho}$ without crossing it, and $\frac{1+\rho}{1+\kappa+\rho}>y\notin N$). Suppose that every subsequent user shares a true story if $y_n<y$ and a false story if $y_n>y$. In the first case $y_n<y<\frac{1+\rho}{1+\kappa+\rho}$, so $y_k$ increases monotonically toward $\frac{1+\rho}{1+\kappa+\rho}>y$; in the second case $y_n>y>\frac{1}{1+\kappa+\rho}$, so $y_k$ decreases monotonically toward $\frac{1}{1+\kappa+\rho}<y$. In either case, since per-period increments are below $\epsilon$, the trajectory enters $B_\epsilon(y)$ at some finite period; let $n+T$ be the first such period. Up to period $n+T$ the trajectory remains in the interval between $y_n$ and $y$, and since $y_n\notin N$, $y\notin N$, and $N$ is the leftmost region, this interval is disjoint from $N$. 
At any state outside $N$ where $y < 1$, there is positive probability of drawing and sharing both true and false stories (at the threshold adjacent to region $N$ this is true by the strictly interior mixing weight in Definition \ref{PGPU definition}). Thus, the required $T$ consecutive shares occur with positive probability (recall that exogenous additions ensure $y_n\in(0,1)$ for $n\ge1$).
\end{proof}

\begin{lemma}\label{lemma trapping}
Let $\{z_n\}$ be a PGPU, let $F$ be its limit differential inclusion, and let $\psi\in(0,1)$ and $\epsilon>0$ be such that $\sign(x)=\sign(\psi-y)$ for all $y\in B_\epsilon(\psi)\setminus\{\psi\}$ and all $x\in F(y)$. Then there exists $m_\psi\in\mathbb{N}$ such that for every $n$ and every state $z$ with $|z|>m_\psi$ and $z^1/|z|\in B_{\epsilon/2}(\psi)$,
\begin{equation*}
\mathbb{P}(y_m\in B_\epsilon(\psi)\ \forall m\ge n \mid z_n=z)\ge 1/2.
\end{equation*}
\end{lemma}

\begin{proof}
Since $\{z_n\}$ is a time-homogeneous Markov chain, it suffices to consider $n=0$ and a fixed initial state $z_0=z$ with $y_0\in B_{\epsilon/2}(\psi)$ and $|z|>m_\psi$. Because at least one ball is added every period, $|z_k|\ge|z|+k$ for all $k$. 

By \eqref{eqn:PGPU decomposition}, $y_{k+1}-y_k=\frac{1}{|z_k|}(g(y_k)+\xi_{k+1}+R_k)$, where, as in the proof of Lemma \ref{lemma perturbed solution}, $\mathbb{E}[\xi_{k+1}\mid\mathcal{F}_k]=0$, $|\xi_{k+1}|\le 4m$ for $m$ the maximal number of balls added per period, $|R_k|\le K_0/|z_k|$, and $g(y)\in F(y)$ for all $y$. Since at most $m$ balls are added per period, each period the share moves by at most $m/|z_k|$.

Let $M_0=0$ and $M_j=\sum_{k=0}^{j-1}\xi_{k+1}/|z_k|$ for $j\ge1$. 
Then $\{M_j\}$ is a martingale with $\sup_j\mathbb{E}[M_j^2]=\sum_{k=0}^{\infty}\mathbb{E}[(\xi_{k+1}/|z_k|)^2]\le 16m^2\sum_{k=0}^{\infty}(|z|+k)^{-2}\le 16m^2/(|z|-1)$, and similarly $\sum_{k=0}^{\infty}|R_k|/|z_k|\le K_0\sum_{k=0}^{\infty}(|z|+k)^{-2}\le K_0/(|z|-1)$.
Choose $m_\psi$ large enough that $|z|>m_\psi$ implies: (a) $m/|z|<\epsilon/16$; (b) $K_0/(|z|-1)<\epsilon/16$; and (c) $(16/\epsilon)^2\cdot 16m^2/(|z|-1)\le 1/2$. By Doob's maximal inequality and (c), the event $G=\{\sup_{j\ge0}|M_j|<\epsilon/16\}$ satisfies $\mathbb{P}(G)\ge 1/2$.

We claim that on $G$ we have $y_k\in B_\epsilon(\psi)$ for all $k\ge0$, which proves the lemma. If not, let $\tau=\min\{k:|y_k-\psi|\ge\epsilon\}$, and let $\sigma=\max\{k<\tau:|y_k-\psi|\le\epsilon/2\}$, which is well defined because $y_0\in B_{\epsilon/2}(\psi)$. Assume w.l.o.g.\ that $y_{\sigma+1}>\psi$. By the definitions of $\sigma$ and $\tau$, $\epsilon/2<|y_k-\psi|<\epsilon$ for all $\sigma<k<\tau$, and by (a) all these $y_k$ lie on the same side of $\psi$ as $y_{\sigma+1}$, since crossing from one side of $B_{\epsilon/2}(\psi)$ to the other would require a single-period move larger than $\epsilon$. So $y_k\in(\psi+\epsilon/2,\psi+\epsilon)$, where $g(y_k)<0$ by hypothesis, for all $\sigma<k<\tau$. Summing the increments from period $\sigma+1$ to period $\tau$ on the event $G$,
\begin{equation*}
y_\tau=y_{\sigma+1}+\sum_{k=\sigma+1}^{\tau-1}\frac{g(y_k)}{|z_k|}+M_\tau-M_{\sigma+1}+\sum_{k=\sigma+1}^{\tau-1}\frac{R_k}{|z_k|}<\psi+\frac{\epsilon}{2}+\frac{\epsilon}{16}+\frac{2\epsilon}{16}+\frac{\epsilon}{16}<\psi+\epsilon,
\end{equation*}
where the inequality uses $y_{\sigma+1}\le y_\sigma+\epsilon/16\le\psi+\epsilon/2+\epsilon/16$ by (a), that every drift term is negative, that $|M_\tau-M_{\sigma+1}|<2\epsilon/16$ on $G$, and (b). Moreover, $y_{\tau-1}\ge\psi-\epsilon/2$ because $y_{\tau-1}$ either equals $y_\sigma$ or lies in $(\psi+\epsilon/2,\psi+\epsilon)$, so $y_\tau>\psi-\epsilon/2-\epsilon/16>\psi-\epsilon$ by (a). Thus $|y_\tau-\psi|<\epsilon$, contradicting the definition of $\tau$.
\end{proof}

   \begin{lemma}\label{claim stable} Assume that $y^*_N \notin N$ and/or  $y_0\notin N$. If $\psi$ is a stable steady state,  there is positive probability that $y_n\rightarrow\psi$.
   \end{lemma}

\begin{proof}
Let $\psi$ be a stable steady state, and pick $\epsilon>0$ small enough that $B_\epsilon(y^*_R)\subset R$ if $\psi=y^*_R$ for some region $R$, and $B_\epsilon(\hat{y}_I)\subset L\cup\{\hat{y}_I\}\cup Q$ if $\psi=\hat{y}_I$, where $L$ and $Q$ are the regions adjacent to $\hat{y}_I$ from the left and from the right.

\noindent\emph{Step 1: An auxiliary process.} Define a PGPU $\{z_{n;\psi}\}$, with share of true stories $y_{n;\psi}=z^1_{n;\psi}/|z_{n;\psi}|$ and limit differential inclusion $F_\psi$ as in Definition \ref{inclusion definition}, as follows. If $\psi=y^*_R$ for a region $R$, let $\{z_{n;\psi}\}$ be the process $\{z_{n;R}\}$ defined in \eqref{lom z_n}, so that $F_\psi(y)=\{g_R(y)\}$ for all $y$. If $\psi=\hat{y}_I$, let $\{z_{n;\psi}\}$ follow the law of motion of $L$ and $Q$ in those regions, while in the third region $O$ exactly one story is added each period---a false story if $O$ lies to the right of $Q$, and a true story if $O$ lies to the left of $L$---so that in $O$ the share moves deterministically toward the nearest other region.

We claim that $\sign(x)=\sign(\psi-y)$ for every $y\ne\psi$ and every $x\in F_\psi(y)$, and that $\psi$ is a steady state of $F_\psi$; together these imply that $\psi$ is the unique steady state of $F_\psi$. If $\psi=y^*_R$, both parts hold by Lemma \ref{lemma unique steady state}, since global stability of the one-dimensional ODE $dy/dt=g_R(y)$ implies $\sign(g_R(y))=\sign(y^*_R-y)$ for all $y\ne y^*_R$. If $\psi=\hat{y}_I$, then since $\hat{y}_I$ is a stable steady state of $F$ we have $y^*_Q<\hat{y}_I<y^*_L$, so by Lemma \ref{lemma unique steady state}, $g_L>0$ on $L$ and at $\hat{y}_I$, and $g_Q<0$ on $Q$ and at $\hat{y}_I$; in $O$ the drift is $-y<0$ if $O$ lies to the right of $Q$ and $1-y>0$ if it lies to the left of $L$; and at each threshold $F_\psi(y)$ is the interval spanned by the two adjacent drifts. The sign condition follows for every $y\ne\hat{y}_I$, and $0\in F_\psi(\hat{y}_I)$ because $g_Q(\hat{y}_I)<0<g_L(\hat{y}_I)$.

Since $\psi$ is the unique steady state of $F_\psi$, by Theorem \ref{thm limit chain transitive}, $y_{n;\psi}$ converges to $\psi$ almost surely. By the sign condition, Lemma \ref{lemma trapping} yields $m_\psi\in\mathbb{N}$ such that for every $n$ and every state $z$ with $|z|>m_\psi$ and $z^1/|z|\in B_{\epsilon/2}(\psi)$, $\mathbb{P}(y_{m;\psi}\in B_\epsilon(\psi)\ \forall m\ge n\mid z_{n;\psi}=z)\ge 1/2$.

\noindent\emph{Step 2: Positive probability of convergence after arriving near $\psi$.} By the choice of $\epsilon$ and the construction of $\{z_{n;\psi}\}$, the transition kernels of $\{z_n\}$ and $\{z_{n;\psi}\}$ coincide at every state whose share lies in $B_\epsilon(\psi)$. Hence, conditional on $z_n=z$ with $|z|>m_\psi$ and $z^1/|z|\in B_{\epsilon/2}(\psi)$, the two processes can be coupled to coincide until the share first exits $B_\epsilon(\psi)$. By Step 1, with probability at least $1/2$ no exit occurs; on this event the processes coincide forever, and since $y_{m;\psi}$ converges to $\psi$ almost surely, $y_m\rightarrow\psi$.

\noindent\emph{Step 3: Positive probability of arriving near $\psi$.} It remains to show that with positive probability the system arrives at a state $z_n$ with $|z_n|>m_\psi$ and $y_n\in B_{\epsilon/2}(\psi)$. By \eqref{general ODE}, for any region $R$,
\[y^*_R = \frac{1+p^T_R(y^*_R)\rho}{1+\kappa+\rho(p^T_R(y^*_R)+p^F_R(y^*_R))}.\]
This implies that $1/(1+\kappa+\rho)<y^*_R<(1+\rho)/(1+\kappa+\rho)$: the first inequality is immediate and the second is equivalent to $\rho(\kappa(1-p^T_R(y^*_R))+p^F_R(y^*_R)(1+\rho))>0$, which always holds. Since any stable steady state is either a quasi steady state or a threshold bounded above and below by quasi steady states,
\begin{equation}\label{y_R inequality}
\frac{1}{1+\kappa+\rho}<\psi<\frac{1+\rho}{1+\kappa+\rho} \quad \forall \psi\in\mathcal{S}.
\end{equation}

First assume $y_0\notin N$. If $\psi\notin N$, then as in the proof of Lemma \ref{lemma connected}, the system reaches a state with $|z_n|>m_\psi$ and $y_n\notin N$ with positive probability, and the claim follows from \eqref{y_R inequality} and Lemma \ref{lemma connected}, together with the fact that $|z_n|$ is nondecreasing. If $\psi\in N$ (so that $\psi=y^*_N$), an argument analogous to the proof of Lemma \ref{lemma connected} shows that consecutive false-story shares, which have positive probability at any state outside $N$, drive the share into $N$ in finite time; from that point on the system never leaves $N$ and converges deterministically to $y^*_N=\psi$. Now assume $y_0\in N$. Then by hypothesis $y^*_N\notin N$, so in region $N$ the share moves deterministically toward $y^*_N\notin N$ and exits $N$ in finite time, after which the previous case applies.
\end{proof}

\begin{lemma}\label{claim repelling}
The system almost surely does not converge to a repelling steady state. 
\end{lemma}
\begin{proof}
Since by Lemma \ref{lemma unique steady state} all quasi steady states are 
stable for their associated ODEs, the only possible repelling steady states 
for the LDI are the thresholds $\hat{y}_I,\hat{y}_M$. Let $\hat{y}$ be a 
repelling steady state. Partition into two complementary events:

\medskip\noindent\textbf{Event $A = \{y_n \in N \text{ infinitely often}\}$.}

If $\hat{y}$ is not adjacent to $N$, then $y_n \in N$ i.o.\ prevents convergence to $\hat{y}$. 
If $\hat{y}$ is adjacent to $N$: since $\hat{y}$ is repelling, the repelling property 
forces $y^*_N \in N$ and $y^*_N \neq \hat{y}$. When $y_n$ enters $N$, the dynamics are 
deterministic and $y_n \to y^*_N$, precluding convergence to $\hat{y}$. Thus 
$\mathbb{P}(A \cap \{y_n \to \hat{y}\}) = 0$.

\medskip\noindent\textbf{Event $A^C = \{y_n \in N \text{ finitely often}\}$.}

To show $\mathbb{P}(A^C\cap\{y_n\to\hat{y}\})=0$, we apply Theorem \ref{thm nonconvergence}, which requires that $\mathbb{E}[\xi_{n+1}^+\mid\mathcal{F}_n]\ge r>0$ for all $y_n$ near $\hat{y}$. However, when $\hat{y}$ is adjacent to $N$, the noise terms $\xi_n$ vanish on the $N$ side of $\hat{y}$.  We therefore apply the theorem to an auxiliary process that has nontrivial noise inside $N$. Write $A^C=\bigcup_{m\in\mathbb{N}}C_m$ where $C_m:=\{y_n\notin N \text{ for all } n\ge m\}$; it suffices to show that $\mathbb{P}(C_m\cap\{y_n\to\hat{y}\})=0$ for every $m$.

For small $\epsilon>0$, let $U:=(\hat{y}-\epsilon,\hat{y}+\epsilon)$ intersect only the two regions adjacent to $\hat{y}$; when $\hat{y}$ is adjacent to $N$, take $\epsilon$ small enough that also $U\subset(y^*_N,1)$, which is possible since $y^*_N\in N$ as noted above. Fix $m$ and let $\{z'_n\}_{n\ge m}$ be an auxiliary PGPU with $z'_m=z_m$, where $y'_n$, $\mathcal{F}'_n$, and $\xi'_n$ are defined as for $\{z_n\}$. Its transition probabilities agree with those of $\{z_n\}$ outside $N$; inside $N$, they are such that the right-hand side of its limit ODE in $N$ is negative on $N\cap U$ and $\mathbb{E}[{\xi'}^+_{n+1}\mid\mathcal{F}'_n]$ is bounded below by a positive constant whenever $y'_n\in N\cap U$, for all sufficiently large $n$ (e.g., adding $w_2=(1,\kappa+\rho)$ or $w_3=(1,\kappa)$ with probability $\tfrac12$ each). Since the transition probabilities of the two processes coincide outside $N$, we can couple them so that $z'_n=z_n$ until the first time $n\ge m$ at which $y_n\in N$. On $C_m$ no such time exists, so $z'_n=z_n$ for all $n\ge m$, and hence $C_m\cap\{y_n\to\hat{y}\}\subseteq\{y'_n\to\hat{y}\}$. It therefore suffices to show that $\mathbb{P}(y'_n\to\hat{y})=0$, which we do by verifying the hypotheses of Theorem \ref{thm nonconvergence} for $\{z'_n\}$: (i) $\hat{y}$ is a repelling steady state of its limit differential inclusion, and (ii) the noise bound holds near $\hat{y}$.

For (i), let $F'$ be the limit differential inclusion of $\{z'_n\}$. Outside the closure of $N$, $F'$ coincides with $F$, and on $N\cap U$ it is the singleton given by the auxiliary limit ODE, which is negative by construction. Since $\hat{y}$ is repelling for $F$ and $N\cap U$, when nonempty, lies to the left of $\hat{y}$, it follows that $\sign(x)=-\sign(\hat{y}-y)$ for all $y\in U\setminus\{\hat{y}\}$ and all $x\in F'(y)$, and that $0\in F'(\hat{y})$. Hence $\hat{y}$ is a repelling steady state of $F'$. For (ii), with $\Delta_T > \Delta_O > \Delta_F$ the three increments from \eqref{lom y_n},
$\mathbb{E}[\xi'^+_{n+1} \mid \mathcal{F}'_n] \geq p^T_R(y'_n)(1 - p^T_R(y'_n))
(\Delta_T - \Delta_O)|z'_n|$  when $y'_n\in U$ is in a region $R\ne N$. Since $p^T_R \in \{y'_n, y'_n/2\}$ 
is bounded away from $0$ and $1$ on $U$, and 
$(\Delta_T - \Delta_O)|z'_n| \geq (1 - y'_n)\rho/4$ for large $|z'_n|$, this is 
bounded below by some $r > 0$. When $y'_n=\hat y$, the probability of sharing a true story lies in $(0,1)$ by  Definition \ref{PGPU definition}, so $\mathbb{E}[\xi'^+_{n+1}\mid\mathcal{F}'_n]$ is again bounded below by a positive constant. On $N \cap U$ the bound holds by construction.
Theorem \ref{thm nonconvergence} yields $\mathbb{P}(y'_n \to \hat{y}) = 0$.  \end{proof}

\begin{proof}[Proof of Lemma \ref{lemma welfare}]
        By Lemma \ref{lemma thresholds}, $V(y,I)$ and $V(y,M)$ are strictly increasing in $y$. Also, $\mathbb{P}(M|y)=\frac{1}{2}y+(1-\delta)(1-y)$ is strictly increasing in $y$, so it suffices to prove that $\mathbb{P}(I|y)V(y,I)$ is strictly increasing in $y$. This will imply that $W(y)$ is strictly increasing for $y\notin N$ because outside of region $N$ at least one of $V(y,I)$ and $V(y,M)$ is strictly positive.
    
    Let $H(y):=\mathbb{P}(I|y)V(y,I)$, then, using the expression for $V(y,I)$ in the proof of Lemma \ref{lemma thresholds} we get, 
\[
\begin{split}
H(y)&= \frac{y+2(1-y)\delta}{2}V(y,I) \\
 &= \frac{1}{2}\left[(1+\lambda)y-2(\mu-1-\lambda)(1-y)\delta \theta  +\frac{1}{\beta}\frac{\left((\mu-1-\lambda)(1-y)\delta \theta\right)^2}{y+2(1-y)\delta}\right].
\end{split}
\]
So,
\[
\begin{split}
 H'(y) &= \frac{1}{2}\left[1+\lambda+2(\mu-1-\lambda)\delta \theta-\frac{\left( (\mu-1-\lambda)\delta \theta\right)^2 (1-y)(1+y+2(1-y)\delta)}{\beta\left( y+2(1-y)\delta\right)^2}\right],  \\
 H'(0) &= \frac{1}{2}\left[1+\lambda+(\mu-1-\lambda)\theta\left(2\delta -\frac{(\mu-1-\lambda)\theta(1+2\delta)}{4\beta}\right)\right], \\
 H''(y) &= \frac{\left((\mu-1-\lambda)\delta \theta\right)^2}{\beta\left( y+2(1-y)\delta\right)^3}.
\end{split}
\]
Note that $H'(0)>0$ since, by Assumption  \ref{sufficient for a below 1} and $\delta>\frac{1}{2}$,
\[
2\delta -\frac{(\mu-1-\lambda)\theta(1+2\delta)}{4\beta} >2\delta -\frac{1+2\delta}{2}>0.
\]
This together with $H''(y)>0$ implies $H'(y)>0$ for all $y\in[0,1]$, which completes the proof. 
\end{proof}

\section{Urn Models}\label{appendix:stochastic approximation}
This appendix extends results from \textcite{schreiber2001urn} and \textcite{benaim2004generalized} about \emph{Generalized Polya urns} (GPUs).   A key feature of these urn models is that the number of balls added each period is bounded,  so that as the overall number of balls grows the change in the system's composition between any two consecutive periods becomes arbitrarily small. Within each of the regions $\{N,I,M,S\}$, our system behaves like a GPU. To analyze the entire system, we define \emph{Piecewise Generalized Polya Urns} (PGPUs), and then combine  results on GPUs with results from  BHS that extend the theory of stochastic approximation to cases where the continuous system is given by a solution to a differential inclusion rather than a differential equation.\footnote{We use results on stochastic approximation for differential inclusions, since existing results on stochastic approximation by differential equations do not apply to the  to the discontinuous system we have here. Extending published results on stochastic approximation to discontinuous differential equations would have been more complicated, and the \citet{hill1980strong} analysis of discontinuous Polya urns only covers the case where a single ball is added each period.}  Theorem \ref{thm limit chain transitive} relates the limit behavior of a PGPU to the limit behavior of the associated differential inclusion; we use it in  the proof of Theorem \ref{thm limit behavior}. Section \ref{appendix:limit ODE} explains why the processes $\{z_{n;R}\}$ defined in \eqref{lom z_n} are GPUs and derives the corresponding limit ODEs.  Section \ref{section:repelling} then proves a result about  repelling steady states for limit inclusions  that is used in the proof of Theorem \ref{thm limit behavior}.
\subsection{Definitions and Notation}
Given a vector $w \in \mathbb{R}^2$ define $\lvert w \rvert = |w^1| +|w^2|$. Let  $\{z_n\}= \{\left(z^1_n,z^2_n\right)\}$ be a homogeneous Markov chain with state space $\mathbb{Z}^2_{+}$ ($\mathbb{Z}_+$ are all the non-negative integers). Let $\Pi:\mathbb{Z}^2_{+}\times \mathbb{Z}^2_{+}\rightarrow[0,1] $ denote its transition kernel, $\Pi(z,z') = \mathbb{P}(z_{n+1}=z'|z_n=z)$. We interpret the  process as an urn model, with $z^i_n$ the number of balls of color $i$ at time step $n$. We now define two types of stochastic processes. 
\begin{definition}\label{GPU definition}
   A Markov process $\{z_n\}$ as above is a \textit{generalized Polya urn} (GPU) if:
   \begin{enumerate}[i.]
    \item Balls are never removed, and in every period at least one and at most $m$ balls are added. Formally, for all $z_n\in \mathbb{Z}^2_{+}$ and all $z_{n+1}$ with $\Pi(z_n,z_{n+1})>0$, we have $z^1_{n+1}\ge z^1_{n}$, $z^2_{n+1}\ge z^2_n$, and $1\le |z_{n+1}-z_n|\le m$.
\item For each $w\in \mathbb{Z}^2_{+}$ with $|w| \leq m$ there  exist Lipschitz-continuous maps $p^w:[0,1] \rightarrow [0,1]$ and a real number $a>0$ such that: $\left|p^w\left(\frac{z^1}{|z|}\right)-\Pi(z,z+w)\right|\le \frac{a}{|z|}$
for all nonzero $z\in \mathbb{Z}_+^2$.
\end{enumerate}
\end{definition}

Let $y_n = \frac{z^1_n}{|z_n|}$ be the share of balls of color $1$ (i.e., of true stories.)   
\begin{definition}\label{stoc. approximation definition}
    Let $\{x_n\}$ be a stochastic process in $[0,1]$ adapted to a filtration $\mathcal\{\mathcal{F}_n\}$. We say that $\{x_n\}$ is a (one dimensional) stochastic approximation if for all $n\in\mathbb{N}$:
    \begin{equation}\label{stochastic approximation equation}
        x_{n+1}-x_n = \gamma_n\left(g(x_n)+\xi_{n+1}+R_n \right),
    \end{equation}
where $\gamma_n\in\mathcal{F}_n$ are non-negative with $\gamma_n\rightarrow 0,\sum_n \gamma_n=\infty$ almost surely, $g$ is a Lipschitz function on $\mathbb{R}$, $\mathbb{E}[\xi_{n+1}| \mathcal{F}_n]=0$ and the remainder terms $R_n\in\mathcal{F}_n$ go to zero and satisfy $\sum_{n=1}^\infty \frac{|R_n|}{n}<\infty$ almost surely. 
\end{definition}
The function $g$ in \eqref{stochastic approximation equation} is the right hand side of the \emph{limit ODE}, $\frac{dx}{dt}=g(x)$. \textcite{schreiber2001urn} and \textcite{benaim2004generalized}  derive the limit ODE of a GPU and prove that with this limit ODE the sequence $\{y_n\}$ of the share of balls of color $1$ is a stochastic approximation process. Since we will later consider a system that includes several GPUs we introduce the notation $\{z_{n;k}\}$ to refer to a general GPU. 
\begin{definition}\label{limit ODE definition}
    For a GPU $\{z_{n;k}\}$ with corresponding maps $p^w_k$, the corresponding \textit{limit ODE} is  $\frac{dy}{dt} = g_k(y)$ where $g_k:[0,1]\rightarrow[0,1]$ is given by\footnote{ Definition \ref{GPU definition}.i implies that only a finite number of the summands are non-zero.} 
    \begin{equation}\label{BTS ODE}
        g_k(y) =\sum_{w \in \mathbb{Z}^2} p^w_k(y)\left(w^1-y|w|\right).
    \end{equation} 
\end{definition}
\subsection{Stochastic Approximation of PGPUs}

This section extends the literature on GPUs to concatenations of GPUs.
\begin{definition}\label{PGPU definition}
  A Markov process $\{z_n\}$ with transition kernel $\Pi$ is a \textit{piecewise generalized Polya urn} (PGPU) if there exists  a finite integer $K$, a finite number of GPUs $\{\{z_{n;k}\}\}_{k=1}^K$ (each with kernel $\Pi_k$), and an interval partition $\{I_k\}_{k=1}^K$ of $[0,1]$, such that for all $z'$, if $\frac{z^1}{|z|}\in \mathring{I}_k$ then $\Pi(z,z') = \Pi_k(z,z')$, where $\mathring{I}$ denotes the interior of $I$.\footnote{At a boundary point $\hat y$ between intervals $I_k$ and $I_{k+1}$, we assume the kernel is  $\alpha_{\hat y}\Pi_k+(1-\alpha_{\hat y})\Pi_{k+1}$  for some fixed $\alpha_{\hat y}\in(0,1)$.}
\end{definition}
The next definition defines the analog of a  limit ODE for a PGPU.
\begin{definition}\label{inclusion definition}
For a PGPU $\{z_n\}$ the \textit{limit differential inclusion} is $\frac{dy}{dt}\in F(y)$, where 
    \begin{equation*}
    F(y) = \begin{cases}
\{g_k(y)\},& y\in\mathring{I}_k\\ 
 \{g_1(0)\},& y=0\\
 \{g_K(1)\}& y=1 \\
[\min\{g_k(y),g_{k+1}(y)\},\max\{g_k(y),g_{k+1}(y)\}], &y=\max(I_k), 1\le k<K\\  
		 \end{cases}
\end{equation*}
\end{definition}
Henceforth, we fix a PGPU $\{z_n\}$ comprised of GPUs $\{\{z_{n;k}\}\}_{k=1}^K$, with share of balls of color $1$ denoted $y_n = \frac{z^1_n}{|z_n|}$ and let 
\begin{equation}\label{general inclusion}
    \frac{dy}{dt} \in F(y)
\end{equation}
be the associated differential inclusion. In order to apply results from BHS, we need to verify that the paper's  standing assumptions on the inclusion hold.\footnote{BHS define inclusions on all of $\mathbb{R}$. Extend $F$ by $F(y)=\{g_1(0)-y\}$ for $y<0$ and $F(y)=\{g_K(1)+1-y\}$ for $y>1$; since every GPU has $g(0)\ge0\ge g(1)$ by \eqref{BTS ODE}, the extension satisfies the standing assumptions below whenever $F$ does and has no steady states outside $[0,1]$, so we do not distinguish between the two.} These are:
\begin{Assumptions_BHS}\label{BHS conditions}
  \begin{enumerate}
    \item  $F$ has a closed graph. 
    \item $F(y)$ is non empty, compact, and convex for all $y\in[0,1]$. 
    \item There exists $c>0$ such that for all $y\in [0,1]$, $
        \sup_{x\in F(y)}|x|\le c(1+|y|). $
\end{enumerate}
 \end{Assumptions_BHS}
 \begin{lemma}\label{lemma:BHS conditions}
The inclusion \eqref{general inclusion} satisfies the standing assumptions in BHS.
 \end{lemma}
 \begin{proof}
 Assumptions 1 and 2  follow immediately from Definition \ref{inclusion definition}. Assumption 3 follows from the fact that the $g_k(y)$ are continuous functions defined over compact sets.  \end{proof}

We relate the limiting behavior of $y_n$ to the solutions to the differential inclusion \eqref{general inclusion} using the ideas of a  \emph{perturbed solution} and a \emph{piecewise affine interpolation.}

\begin{definition}
A continuous function $\textbf{Y}:[0,\infty)\rightarrow \mathbb{R}$ is a \emph{perturbed solution} to $F$ if it is absolutely continuous and there are a locally integrable $t\mapsto U(t)$ and a $\delta:[0,\infty)\rightarrow\mathbb{R}$ with $\lim_{t\rightarrow\infty}\delta(t)=0$ such that $\lim_{t\rightarrow \infty} \sup_{0\le h \le T } |\int_{t}^{t+h}U(s)ds| =0$ for all $T>0$ and $d\textbf{Y}(t)/{dt}-U(t)\in F^{\delta(t)}(\textbf{Y}(t))$ for almost every $t>0$, where $F^{\delta}(y)=\{x: |x-x'|<\delta \text{ for some } x'\in F(z) \text{ with } |z-y|<\delta\}$.
\end{definition}

\begin{definition}
    The \emph{piecewise affine interpolation} of $y_n$ is  
\begin{equation*}
    \textbf{Y}(t) = y_n+\frac{t-\tau_n}{\gamma_{n+1}}(y_{n+1}-y_{n}), \quad t\in[\tau_n,\tau_{n+1}].
\end{equation*}
where $\tau_{0}=0$, $\tau_{n+1} = \tau_n+\frac{1}{|z_n|}$, and $\gamma_{n+1} = \frac{1}{|z_n|}$. 
\end{definition}

The following lemma plays the role for PGPUs that Theorem 2.2 of \textcite{schreiber2001urn} plays for GPUs.

\begin{lemma}\label{lemma perturbed solution}
    Let $\{z_n\}$ be a PGPU and \eqref{general inclusion} its limit differential inclusion, and let \textbf{Y} be its piecewise affine interpolation.  Then, almost surely, \textbf{Y} is a bounded perturbed solution to \eqref{general inclusion}. 
\end{lemma}
\begin{proof}
Define $g:[0,1]\rightarrow\mathbb{R}$ by $g(y)=g_k(y)$ for $y\in\interior(I_k)$ and $g(\hat y)=\alpha_{\hat y} g_k(\hat y)+(1-\alpha_{\hat y})g_{k+1}(\hat y)$ at each boundary point $\hat y$ between $I_k$ and $I_{k+1}$, where $\alpha_{\hat y}$ is the mixing weight from Definition~\ref{PGPU definition}; by Definition \ref{inclusion definition}, $g(y)\in F(y)$ for all $y\in[0,1]$. Let $\xi_n$ be the rescaled martingale increments defined before Theorem \ref{thm nonconvergence}, and let $R_n=|z_n|\mathbb{E}[y_{n+1}-y_n|z_n]-g(y_n)$. Then $\xi_n,R_n$ are $\mathcal{F}_n$-measurable, $\mathbb{E}[\xi_{n+1}|\mathcal{F}_n]=0$, and
\begin{equation}\label{eqn:PGPU decomposition}
y_{n+1}-y_n=\frac{1}{|z_n|}(g(y_n)+\xi_{n+1}+R_n).
\end{equation}
When $y_n\in\interior(I_k)$, $\{z_n\}$ follows the same law of motion as the GPU $\{z_{n;k}\}$; at a boundary point, it follows a convex combination of the two adjacent GPUs' laws of motion. In either case, Lemma 1 in \textcite{benaim2004generalized} yields $K_0>0$ such that $|R_n|\le K_0/|z_n|$ and $|\xi_n|\le 4m$, where $m$ is the maximal number of balls added per period.

Since $g(y_n)\in F(y_n)$, \eqref{eqn:PGPU decomposition} gives $y_{n+1}-y_n-\gamma_{n+1}U_{n+1}\in \gamma_{n+1}F(y_n)$ with $U_{n+1}=\xi_{n+1}+R_n$ and $\gamma_{n+1}=1/|z_n|$, and $\textbf{Y}$ is the affine interpolation of this recursion on the time scale $\{\tau_n\}$. Because at least one and at most $m$ balls are added each period, $|z_0|+n\le |z_n|\le |z_0|+mn$, so $\gamma_n\rightarrow 0$ and $\sum_n \gamma_n=\infty$. Proposition 1.3 of BHS states that the affine interpolation of such a recursion is a perturbed solution to $F$ provided that (i) for all $T>0$, $\lim_{n\rightarrow\infty}\sup\{|S_k-S_n|: k>n \text{ with } \tau_k\le\tau_n+T\}=0$, where $S_n=\sum_{i=0}^{n-1}\gamma_{i+1}U_{i+1}$, and (ii) $\sup_n|y_n|<\infty$.\footnote{BHS take the step sizes $\gamma_n$ to be deterministic, but Proposition 1.3 is a pathwise statement, so once (i) and (ii) hold almost surely it applies along almost every sample path; here $\gamma_n\rightarrow0$ and $\sum_n\gamma_n=\infty$ hold for every path.}
Condition~(ii) holds since $y_n \in [0,1]$. For condition~(i), write
\[
S_n = M_n + \sum_{i=0}^{n-1} \gamma_{i+1} R_i, 
\quad \text{where} \quad 
M_n = \sum_{i=0}^{n-1} \gamma_{i+1} \xi_{i+1}.
\]
Since $\gamma_{i+1}=1/|z_i|$ is $\mathcal{F}_i$-measurable, $\{M_n\}$ is a martingale, and it is bounded in $L^2$ because $\sum_{i=0}^{\infty}\mathbb{E}[\gamma_{i+1}^2\xi_{i+1}^2]\le 16m^2\sum_{i=0}^{\infty}(|z_0|+i)^{-2}<\infty$. Hence $M_n$ converges almost surely, so $\sup_{k>n}|M_k-M_n|\rightarrow 0$ almost surely. For the remainder terms, $\tau_k\le\tau_n+T$ gives $\sum_{i=n}^{k-1}\gamma_{i+1}=\tau_k-\tau_n\le T$, and $|R_i|\le K_0/(|z_0|+n)$ for all $i\ge n$, so $\sum_{i=n}^{k-1}\gamma_{i+1}|R_i|\le K_0T/(|z_0|+n)\rightarrow 0$. Condition (i) follows, so $\textbf{Y}$ is almost surely a perturbed solution to $F$, and it is bounded since $\textbf{Y}(t)\in[0,1]$.
\end{proof}

We are now ready to state and prove Theorem \ref{thm limit chain transitive}.  The proof combines the previous results with a direct application of the following theorem:
\begin{thm3.6_BHS}
  If $\mathbf{x}$ is a bounded perturbed solution to $F$, the limit set of $\mathbf{x}$, $L(\mathbf{x}) = \bigcap_{t\ge 0}\overline{\{\mathbf{x}(s) : s>t\}}$ is internally chain transitive.\footnote{BHS extend the definition of internal chain transitivity to differential inclusions.}  
\end{thm3.6_BHS}

\begin{theorem}\label{thm limit chain transitive}
    Let $\{z_n\}$ be a PGPU, $\{y_n\}$ the share of balls of color $1$ and $F$ the associated limit differential inclusion. Then the limit set of $\{y_n\}$, $L(y_n) = \bigcap_{m> 0}\overline{\{y_n : n>m\}}$, is almost surely internally chain transitive for $F$. 
\end{theorem}
\begin{proof} By Lemma \ref{lemma perturbed solution}, the interpolation $\textbf{Y}$ is almost surely a perturbed solution to $F$. Note that it is also bounded since $\textbf{Y}(t)\in [0,1]$ for all $t\ge0$. Thus, Theorem 3.6 in BHS  implies that the limit set of $\textbf{Y}$ is almost surely internally chain transitive for $F$. Note that the asymptotic behaviors of $\textbf{Y}(t)$ and $y_n$ are the same by the definition of interpolation, i.e., $L(y_n)=L(\mathbf{Y})$, which completes the proof.

\end{proof}
\vspace{-2em}
\subsection{The GPUs $\{z_{n;R}\}$} \label{appendix:limit ODE}

This section shows that the processes $\{z_{n;R}\}$ as defined in \eqref{lom z_n} are GPUs and derive the formula for their limit ODEs. 
\begin{lemma}
    For each $R\in\{N,I,M,S\}$, $\{z_{n;R}\}$ is a GPU with limit ODE given by  \eqref{general ODE}. 
\end{lemma}
\begin{proof} Let $R$ be one of the four possible regions. To show that $\{z_{n;R}\}$ is a GPU we need to verify the conditions of Definition \ref{GPU definition}. Condition i) follows directly from \eqref{lom z_n}: each increment vector adds at least one and at most $m=1+\kappa+\rho$ balls.  For condition ii), let $w_1=\begin{pmatrix}1+\rho\\\kappa \end{pmatrix}, w_2 = \begin{pmatrix} 1\\ \kappa+\rho \end{pmatrix}, w_3 = \begin{pmatrix}1\\ \kappa\end{pmatrix} $, and let $p^T_R(y),p^F_R(y), 1-p^T_R(y)-p^F_R(y)$ respectively be the  maps $p^w$ corresponding to these vectors. By \eqref{probabilities} all three maps are Lipschitz-continuous. Let $\Pi_R$ denote the transition kernel for $\{z_{n;R}\}$. By the law of motion \eqref{lom z_n}, for any $w\in\{w_1,w_2,w_3\}$ and for any $z\in \mathbb{Z}^2_{+}$: $\Pi_R(z,z+w) = p^w\left(\frac{z^1}{|z|}\right)$. Since $\Pi_R(z,z+w)=0$ for any $w\notin\{w_1,w_2,w_3\}$, condition ii) is satisfied.

Next,   \eqref{probabilities}, \eqref{lom z_n}, and \eqref{BTS ODE} imply that the ODE associated with $\{z_{n;R}\}$ is
\begin{equation*}
    \begin{split}
        g_R(y) &= p^T_R(y)(1+\rho - y(1+\rho+\kappa))+p^F_R(y)(1- y(1+\rho+\kappa))\\
        &+(1-p^T_R(y)- p^F_R(y))(1-y(1+\kappa)).
   \end{split}
\end{equation*}
Rearranging gives
$g_R(y) =  1+ p^T_R(y)\rho -y\left(1+\kappa+\rho\left(p^T_R(y)+p^F_R(y)\right) \right),$
as in \eqref{general ODE}. 
\end{proof}

\subsection{Repelling Steady States}\label{section:repelling}
This subsection shows that if $\psi$ is a repelling steady state for the LDI, then under a condition on the noise in the stochastic system, $\mathbb{P}(y_n\rightarrow\psi)=0$.
Consider a PGPU $\{z_n\}$, comprised of GPUs $\{z_{n;k}\}_{k=1}^K$ with associated intervals $I_k$, where $g_k$ is the RHS of the limit ODE for GPU $\{z_{n;k}\}$. Let $y_{n;k}=\frac{z^1_{n;k}}{|z_{n;k}|}$. Recall that $y_n = \frac{z^1_n}{|z_n|}$ and that the LDI for this PGPU is given by \eqref{general inclusion}. We now add the following assumption, which is satisfied by the PGPUs in our model:
\begin{assumption}\label{assumption globally stable}
    Each limit ODE $\frac{dy}{dt}=g_k(y)$ has a globally stable steady state $y^*_k$. 
\end{assumption}

Assumption \ref{assumption globally stable} implies that the only possible repelling steady states for the LDI are the thresholds between the intervals $I_k$. Define these these as $\hat{y}_k= \max\{I_k\}$ for $k=1,\ldots,K$. Finally, let $\mathcal{F}_n$ be the $\sigma$-algebra generated by $(z_1,...,z_n)$, let $\xi_{n+1} = (y_{n+1}-y_{n}-\mathbb E[y_{n+1}-y_{n}|z_n])|z_n|$ and denote $\xi_n^+=\max\{0,\xi_n\}, \xi_n^-=-\min\{0,\xi_n\}$.

\begin{theorem}\label{thm nonconvergence}
Let $\hat{y}$ be the threshold between intervals $I_k,I_{k+1}$ and assume that $\hat{y}$ is a repelling steady state for the LDI. If there exist $\epsilon, r>0$ such that for all $n\in\mathbb{N}$: $\mathbb{E}[\xi^+_{n+1}|\mathcal{F}_n]\ge r$  if $y_{n}\in (\hat{y}-\epsilon,\hat{y}+\epsilon)$,  then $\mathbb{P}(y_n\rightarrow \hat{y})=0$.   
\end{theorem}
The proof applies the following result:
\begin{thm2.9_Pemantle}
Suppose $\{x_n\}$ is a stochastic approximation process as defined in Definition \ref{stoc. approximation definition} except that $g$ need not be continuous.\footnote{In \textcite{pemantle2007survey} the result is stated for a process with step size $1/n$ but the proof extends to any step size that is $\Theta(1/n)$ almost surely.} Assume that for some $p\in(0,1)$ and $\epsilon>0$: $\sign(g(x))=-\sign(p-x)$ for all $x\in(p-\epsilon,p+\epsilon)$. Suppose further that the martingale terms $\xi_n$ in the stochastic approximation equation \eqref{stochastic approximation equation} are such that $\mathbb{E}[\xi_{n+1}^+|\mathcal{F}_n] \text{ and }\mathbb{E}[\xi_{n+1}^-|\mathcal{F}_n]$ are bounded above and below by positive numbers when $x_n\in(p-\epsilon,p+\epsilon)$. Then $\mathbb{P}(x_n\rightarrow p)=0$. 
\end{thm2.9_Pemantle}

\begin{proof} Let $g$ and $R_n$ be as in the proof of Lemma \ref{lemma perturbed solution}, so that \eqref{eqn:PGPU decomposition} holds with $|R_n|\le K_0/|z_n|$ and $|\xi_n|\le 4m$. This has the form of the stochastic approximation \eqref{stochastic approximation equation}, except that $g$ is generally not continuous. Since $|z_n|\ge|z_0|+n$, we have $\sum_{n=1}^{\infty}|R_n|/n\le K_0\sum_{n=1}^{\infty}1/(n(|z_0|+n))<\infty$, so $\{y_n\}$ is a stochastic approximation. The bound $|\xi_n|\le 4m$ gives $\mathbb{E}[\xi_{n+1}^+|\mathcal{F}_n], \mathbb{E}[\xi_{n+1}^-|\mathcal{F}_n]\le 4m$.

It remains to verify the hypotheses of Pemantle's Theorem 2.9 at $p=\hat y$. The sign condition follows from $\hat y$ being repelling: shrinking $\epsilon$ if necessary, $\sign(g(y)) = -\sign(\hat y - y)$ for $y \in (\hat y - \epsilon, \hat y + \epsilon) \setminus \{\hat y\}$. This outward sign in the punctured neighborhood is sufficient for nonconvergence; any non-zero drift exactly at $\hat{y}$ simply provides an additional force pushing the process away. For the noise condition, the hypothesis gives $\mathbb{E}[\xi_{n+1}^+|\mathcal{F}_n]\ge r$ when $y_n\in(\hat y-\epsilon,\hat y+\epsilon)$. Since $\mathbb{E}[\xi_{n+1}|\mathcal{F}_n]=0$, $\mathbb{E}[\xi_{n+1}^+|\mathcal{F}_n]=\mathbb{E}[\xi_{n+1}^-|\mathcal{F}_n]$, so the same lower bound holds for $\mathbb{E}[\xi_{n+1}^-|\mathcal{F}_n]$. Pemantle's Theorem 2.9 then yields $\mathbb{P}(y_n\rightarrow\hat y)=0$. \end{proof}

\setstretch{1}
\printbibliography

@article{benaim2004generalized,
  title={Generalized urn models of evolutionary processes},
  author={Benaim, Michel and Schreiber, Sebastian J and Tarres, Pierre},
  journal={Annals of Applied Probability},
  pages={1455--1478},
  year={2004},
  publisher={JSTOR}
}

@article{hill1980strong,
  title={A strong law for some generalized urn processes},
  author={Hill, Bruce M and Lane, David and Sudderth, William},
  journal={The Annals of Probability},
  pages={214--226},
  year={1980},
}

@article{benaim2005stochastic,
  title={Stochastic approximations and differential inclusions},
  author={Benaim, Michel and Hofbauer, Josef and Sorin, Sylvain},
  journal={SIAM Journal on Control and Optimization},
  volume={44},
  number={1},
  pages={328--348},
  year={2005},
  publisher={SIAM}
}

@article{pennycook2021shifting,
  title={Shifting attention to accuracy can reduce misinformation online},
  author={Pennycook, Gordon and Epstein, Ziv and Mosleh, Mohsen and Arechar, Antonio A and Eckles, Dean and Rand, David G},
  journal={Nature},
  volume={592},
  number={7855},
  pages={590--595},
  year={2021},
  publisher={Nature Publishing Group UK London}
}

@article{pennycook2018prior,
  title={Prior exposure increases perceived accuracy of fake news.},
  author={Pennycook, Gordon and Cannon, Tyrone D and Rand, David G},
  journal={Journal of Experimental Psychology: General},
  volume={147},
  number={12},
  pages={1865},
  year={2018},
  publisher={American Psychological Association}
}

@article{allcott2017social,
  title={Social media and fake news in the 2016 election},
  author={Allcott, Hunt and Gentzkow, Matthew},
  journal={Journal of Economic Perspectives},
  volume={31},
  number={2},
  pages={211--236},
  year={2017},
  publisher={American Economic Association 2014 Broadway, Suite 305, Nashville, TN 37203-2418}
}

@article{schreiber2001urn,
  title={Urn models, replicator processes, and random genetic drift},
  author={Schreiber, Sebastian J},
  journal={SIAM Journal on Applied Mathematics},
  volume={61},
  number={6},
  pages={2148--2167},
  year={2001},
  publisher={SIAM}
}

@book{mahmoud2008polya,
  title={P{\'o}lya urn models},
  author={Mahmoud, Hosam},
  year={2008},
  publisher={CRC Press}
}

@article{vosoughi2018spread,
  title={The spread of true and false news online},
  author={Vosoughi, Soroush and Roy, Deb and Aral, Sinan},
  journal={Science},
  volume={359},
  number={6380},
  pages={1146--1151},
  year={2018},
  publisher={American Association for the Advancement of Science}
}

@article{chen2023makes,
  title={What makes news sharable on social media?},
  author={Chen, Xi and Pennycook, Gordon and Rand, David},
  journal={Journal of Quantitative Description: Digital Media},
  volume={3},
  year={2023}
}

@article{acemoglu2023online,
    author = {Acemoglu, Daron and Ozdaglar, Asuman and Siderius, James},
    title = "{A model of online misinformation}",
    journal = {The Review of Economic Studies},
volume={91},
  pages={3117--3150},
    year = {2024},
   }

@article{dasaratha2022learning,
  title={Learning from viral content},
  author={Dasaratha, Krishna and He, Kevin},
  journal={arXiv preprint arXiv:2210.01267},
  year={2023}
}

@article{papanastasiou2020fake,
  title={Fake news propagation and detection: A sequential model},
  author={Papanastasiou, Yiangos},
  journal={Management Science},
  volume={66},
  number={5},
  pages={1826--1846},
  year={2020},
  publisher={INFORMS}
}

@article{pemantle2007survey,
author = {Robin Pemantle},
title = {{A survey of random processes with reinforcement}},
volume = {4},
journal = {Probability Surveys},
publisher = {Institute of Mathematical Statistics and Bernoulli Society},
pages = {1--79},
year = {2007}
}

@article{pennycook2020implied,
  title={The implied truth effect: Attaching warnings to a subset of fake news headlines increases perceived accuracy of headlines without warnings},
  author={Pennycook, Gordon and Bear, Adam and Collins, Evan T and Rand, David G},
  journal={Management Science},
  volume={66},
  number={11},
  pages={4944--4957},
  year={2020},
  publisher={INFORMS}
}

@unpublished{mostagir2022naive,
  title={Naive and {B}ayesian learning with misinformation policies},
  author={Mostagir, Mohamed and Siderius, James},
  note= {Working paper, University of Michigan and Massachusetts Institute of Technology},
  year={2022},
 }

@article{kranton2024social,
  title={Social connectedness and information markets},
  author={Kranton, Rachel and McAdams, David},
  journal={American Economic Journal: Microeconomics},
  volume={16},
  number={1},
  pages={33--62},
  year={2024},
  publisher={American Economic Association 2014 Broadway, Suite 305, Nashville, TN 37203-2425}
}

@article{bloch2018rumors,
  title={Rumors and social networks},
  author={Bloch, Francis and Demange, Gabrielle and Kranton, Rachel},
  journal={International Economic Review},
  volume={59},
  number={2},
  pages={421--448},
  year={2018},
  publisher={Wiley Online Library}
}

@article{merlino2023debunking,
  title={Debunking rumors in networks},
  author={Merlino, Luca P and Pin, Paolo and Tabasso, Nicole},
  journal={American Economic Journal: Microeconomics},
  volume={15},
  number={1},
  pages={467--496},
  year={2023},
  publisher={American Economic Association 2014 Broadway, Suite 305, Nashville, TN 37203-2425}
}

@article{arieli2024sequential,
  title={Sequential naive learning},
  author={Arieli, Itai and Babichenko, Yakov and Mueller-Frank, Manuel},
  journal={Available at SSRN 3753401},
  year={2024}
}

@article{guriev2023curtailing,
  title={Curtailing False News, Amplifying Truth},
  author={Guriev, Sergei and Henry, Emeric and Marquis, Th{\'e}o and Zhuravskaya, Ekaterina},
  journal={Econometrica},
  note={Forthcoming},
  year={2026}
}

@article{pennycook2022accuracy,
  title={Accuracy prompts are a replicable and generalizable approach for reducing the spread of misinformation},
  author={Pennycook, Gordon and Rand, David G},
  journal={Nature Communications},
  volume={13},
  number={1},
  pages={2333},
  year={2022},
  publisher={Nature Publishing Group UK London}
}

@unpublished{smith2020rational,
  title={Rational Social Learning with Random Sampling},
  author={Smith, Lones and S{\o}rensen, Peter Norman},
    note={Working paper, University of Wisconsin and University of Copenhagen},
  year={2020}
}

@article{arthur1993information,
  title={Information contagion},
  author={Arthur, W Brian and Lane, David A},
  journal={Structural Change and Economic Dynamics},
  volume={4},
  number={1},
  pages={81--104},
  year={1993},
  publisher={Elsevier}
}

@article{paul2016russian,
  title={The {R}ussian “firehose of falsehood” propaganda model},
  author={Paul, Christopher and Matthews, Miriam},
  journal={Rand Corporation},
  volume={2},
  number={7},
  pages={1--10},
  year={2016},
  publisher={JSTOR}
}

@inproceedings{zhang2021understanding,
  title={Understanding the diverging user trajectories in highly-related online communities during the {COVID-19} pandemic},
  author={Zhang, Jason Shuo and Keegan, Brian and Lv, Qin and Tan, Chenhao},
  booktitle={Proceedings of the International AAAI Conference on Web and Social Media},
  volume={15},
  pages={888--899},
  year={2021}
}
\onehalfspacing

\newpage

\renewcommand\thesection{C}
\renewcommand{\thepage}{\arabic{page}}
    \setcounter{page}{1}
\section{Online Appendix}\label{appendix:online}
\subsection{Comparative Statics}\label{appendix:comparative}
 This section analyzes comparative statics for the quasi steady states and thresholds with respect to all model parameters. Theorems \ref{thm comparative sharing}, \ref{thm comparative evoc} and \ref{thm comparative boring} present comparative statics for the quasi steady states $y^*_S,y^*_I$ and $y^*_M$ respectively. Theorem \ref{thm comparative thresholds} presents comparative statics for the thresholds $\hat{y}_I,\hat{y}_M$.  Theorem \ref{thm:rho limit} considers the limits of the quasi steady states as $\rho\rightarrow\infty$. The end of the section summarizes the findings that are not discussed in the main text. 
 
\begin{theorem}\label{thm comparative sharing}
The quasi steady state $y^*_S$ is increasing in $\rho$ and $\mu$  and decreasing in $\kappa,\beta$ and $\lambda$. There exists $\theta_S\in(0,1]$ (whose value depends on the other parameters) such that $y^*_S$ is decreasing in $\theta$ for $\theta<\theta_S$ and increasing in $\theta$ for $\theta>\theta_S$. $y^*_S$ is decreasing in $\delta$ for $\delta$ sufficiently close to $\frac{1}{2}$ and increasing in $\delta$ for $\delta$ sufficiently close to 1.
\end{theorem}

\begin{proof}
Let $r_0 = (\rho_0,\kappa_0,\theta_0,\mu_0,\beta_0,\delta_0,\lambda_0)$ be a vector of parameters and consider $g_S(y)$ as a function $G(y,r):\mathbb{R}^8\rightarrow\mathbb{R}$. Let $y^*_0\in(0,1)$ be the unique $y\in[0,1]$ that solves 
\begin{equation}\label{implicit function sharing}
G(y^*_0,r_0)=0.  
\end{equation}
 Lemma \ref{lemma unique steady state} implies that  $G(y,r_0)>0$ for $y<y^*_0$ and $G(y,r_0)<0$ for $y>y^*_0$ so it must be the case that $G_{y}(y^*_0,r_0)\le 0$. Moreover,  it cannot be the case that $G_{y}(y^*_0,r_0)=0$ because that would imply that $y^*_0$ is a local maximum for $G_{y}(\cdot,r_0)$ while the proof of Lemma \ref{lemma unique steady state} shows that the second derivative of this function (the third derivative w.r.t $y$ of $G(y,r_0)$) is strictly positive over $[0,1]$, so $G_y(y^*_0,r_{0})<0$. 
 
 Since $G(y^*_0,r_0)=0$ and $G_{y}(y^*_0,r_0)\neq 0$, by the implicit function theorem \eqref{implicit function sharing} defines a function $y^*_S(r):\mathbb{R}^7\rightarrow \mathbb{R}$ in some neighborhood of $r_0$, such that $y^*_S(r)$ is the unique steady state of the ODE $\frac{dy}{dt}=g_S(y)$ in $[0,1]$, and 
\begin{equation}
    \nabla y^*_S(r_0) = -\frac{1}{G_{y}(y^*_0,r_0)}\nabla_{r}G(y^*_0,r_0)
\end{equation}

Furthermore, since $G_{y}(y^*_0,r_0)< 0$, for all $x\in(\rho,\kappa,\theta,\mu,\beta,\delta,\lambda)$: $\sign(\frac{dy^*_S(r_0)}{dx}) = \sign(G_x(y^*_0,r_0))$. Plugging $p^T_S,p^F_S$ from \eqref{probabilities} into \eqref{general ODE} and rearranging yields
\[
\begin{split}
G(y,r) &=  1+(\rho(1-\theta)-1-\kappa)y-\rho(1-\theta)y^2 \\
&+\frac{\rho y (\theta(1-y))^2}{\beta}\left( \frac{(\mu-1-\lambda)\delta^2}{y+2(1-y)\delta} +\frac{(\mu-1)(1-\delta)^2}{y+2(1-y)(1-\delta)} \right)    
\end{split}
\]

We now solve for the sign of each of the partial derivatives of $G$.

\noindent$\boldsymbol{\rho}$: $G_\rho(y,r) =  (1-\theta)y(1-y)+\frac{ y (\theta(1-y))^2}{\beta}\left( \frac{(\mu-1-\lambda)\delta^2}{y+2(1-y)\delta} +\frac{(\mu-1)(1-\delta)^2}{y+2(1-y)(1-\delta)} \right) >0$

\noindent$\boldsymbol{\theta}$: $G_\theta(y,r) =  \rho y(1-y)\left(\frac{2\theta(1-y)}{\beta}\left( \frac{(\mu-1-\lambda)\delta^2}{y+2(1-y)\delta} +\frac{(\mu-1)(1-\delta)^2}{y+2(1-y)(1-\delta)} \right)-1 \right)$.

So $G_\theta(y,r)>0$ if and only if, 
\[
\theta> \frac{\beta}{2(1-y)}\left(\frac{1}{\frac{(\mu-1-\lambda)\delta^2}{y+2(1-y)\delta} +\frac{(\mu-1)(1-\delta)^2}{\left(y+2(1-y)(1-\delta)\right)}}\right)
\]
Note that the RHS is always positive, so that for sufficiently small $\theta$, $y^*_S$ is decreasing in $\theta$. However, it is possible that the RHS is below 1 so that for large values of $\theta$ the relationship reverses. See Appendix \ref{appendix:additional} for an example.

 \noindent$\boldsymbol{\kappa}$: $G_\kappa(y,r) = -y<0.$

\noindent$\boldsymbol{\mu}$:  $
G_\mu(y,r) =  \frac{\rho y (\theta(1-y))^2}{\beta}\left( \frac{\delta^2}{y+2(1-y)\delta} +\frac{(1-\delta)^2}{\left(y+2(1-y)(1-\delta)\right)} \right)>0.$

\noindent$\boldsymbol{\beta}$: $G_\beta(y,r) = -\frac{\rho y (\theta(1-y))^2}{\beta^2}\left( \frac{(\mu-1-\lambda)\delta^2}{y+2(1-y)\delta} +\frac{(\mu-1)(1-\delta)^2}{y+2(1-y)(1-\delta)} \right)  <0.$

\noindent$\boldsymbol{\delta}$: $G_\delta(y,r) =\frac{2\rho y (\theta(1-y))^2}{\beta}\left[ \frac{(\mu-1-\lambda)\delta\left(y+(1-y)\delta\right)}{\left(y+2(1-y)\delta\right)^2} +\frac{(\mu-1)(1-\delta)((1-y)\delta-1)}{\left(y+2(1-y)(1-\delta)\right)^2} \right]. $  

So, fixing all parameters except $\delta$ we have $\sign G_\delta(y,r) = \sign(s(y,\delta))$, where $s(y,\delta)$ is the expression in square brackets.
Note that $s(y,1/2) = -\frac{(1+y)\lambda}{4}<0$ and $s(y,1) =\frac{\mu-1-\lambda}{(2-y)^2}>0$ so $y^*_S$ is decreasing in $\delta$ for small values of $\delta$ and increasing in $\delta$ for large values of $\delta$ (recall that we assume $\delta\ge\frac{1}{2}$). 

\noindent$\boldsymbol{\lambda}$:  $G_\lambda(y,r) = -\frac{\rho y (\theta(1-y)\delta)^2}{\beta\left(y+2(1-y)\delta\right)} <0.$
\end{proof}

\begin{theorem}\label{thm comparative evoc}
The quasi steady state $y^*_I$ is increasing in $\mu$ and decreasing in $\kappa,\beta,\lambda$ and $\delta$. $y^*_I$ is increasing in $\rho$ if $d_I(y^*_I)>0$ and decreasing in $\rho$ when the sign is reversed, and both cases can arise in  region $I$.\footnote{Part 4 of Theorem \ref{thm:rho limit} shows we can replace the condition $d_I(y^*_I)>0$ here with $d_I(\frac{1}{1+\kappa})$, which is the condition presented in the main text.} There exists $\theta_I\in(0,1]$ (whose value depends on the other parameters) such that $y^*_I$ is decreasing in $\theta$ for $\theta<\theta_I$ and increasing in $\theta$ for $\theta>\theta_I$. 
\end{theorem}

\begin{proof}
By a similar argument as in the proof of Theorem \ref{thm comparative sharing}, for all $x\in(\rho,\kappa,\theta,\mu,\beta,\delta,\lambda)$: $\sign(\frac{dy^*_I(r_0)}{dx}) = \sign(G_x(y^*_0,r_0))$ where now $G(y,r)$ is given by
\begin{equation*}
  G(y,r) =   1+\left(\rho\left(\frac{1}{2}-\delta \theta\right)-1-\kappa\right)y-\rho \left(\frac{1}{2}-\delta \theta\right)y^2 +\rho y (\delta \theta(1-y))^2\frac{\mu-1-\lambda}{\beta(y+2(1-y)\delta)}.
\end{equation*}
We now solve for the sign of each of the partial derivatives of $G$.

\noindent$\boldsymbol{\rho}$: $G_\rho(y,r) = y(1-y)\left[\frac{1}{2}-\delta \theta \left(1-\frac{(1-y)\delta \theta(\mu-1-\lambda)}{\beta(y+2(1-y)\delta)}\right)\right].$

Note that the expression in square brackets is exactly $d_I(y)$, so $\sign(G_\rho(y^*_I,r)) = \sign(d_I(y^*_I))$. In Appendix \ref{appendix:additional} we show that both $d_I(y^*_I)>0$ and $d_I(y^*_I)<0$ are possible and can occur when $y^*_I\in I$.

\noindent$\boldsymbol{\theta}$:  $G_\theta(y,r) = \delta \rho y (1-y)\left(\frac{2 \delta \theta (1-y) (\mu-1-\lambda)}{\beta (y+2 \delta (1-y))}-1\right).$

So, $G_\theta(y,r)>0$ if and only if
\begin{equation*}
\theta>  \frac{\beta (y+2 \delta (1-y))}{2 \delta (1-y) (\mu-1-\lambda)}.
\end{equation*}
Note that the RHS is always positive, so that for sufficiently small $\theta$, $y^*_I$ is decreasing in $\theta$. However, it is possible that the RHS is below 1 so that for large values of $\theta$ the relationship reverses. See Appendix \ref{appendix:additional} for an example.  

\noindent$\boldsymbol{\kappa}$: $G_\kappa(y,r) = -y<0.$

\noindent$\boldsymbol{\mu}$: $G_\mu(y,r) = \frac{\rho y (\delta \theta(1-y))^2}{\beta(y+2(1-y)\delta)}  >0.$

\noindent$\boldsymbol{\beta}$: $G_\beta(y,r) = -\rho y (\delta \theta(1-y))^2\frac{ (\mu-1-\lambda)}{\beta^2(y+2(1-y)\delta)}<0$

\noindent$\boldsymbol{\delta}$: $G_\delta(y,r)= \rho \theta y(1-y) \left[\frac{2 \delta \theta (1-y) \left(y+\delta (1-y)\right) (\mu-1-\lambda)}{\beta \left(y+2 \delta (1-y)\right)^2}-1\right] <0 $.

For the inequality, let $f(y)$ denote the expression in square brackets. It suffices to prove $f(y)<0$ for all $y$. This follows from $f(0) = \frac{\theta (\mu-1-\lambda)}{2\beta}-1<0 $ (by Assumption \ref{sufficient for a below 1}), and $f'(y) =  - \frac{2\delta \theta y(\mu-1-\lambda) }{\beta \left(y+2 \delta (1-y)\right)^3}<0$.

\noindent$\boldsymbol{\lambda}$: $G_\lambda(y,r) =-\frac{\rho y (\delta \theta(1-y))^2}{\beta(y+2(1-y)\delta)} <0$.
\end{proof}

\begin{theorem}\label{thm comparative boring}
The quasi steady state $y^*_M$ is constant in $\lambda$, increasing in $\mu,\rho$, and $\delta$, and decreasing in $\kappa$ and $\beta$. There exists $\theta_M\in(0,1]$ (whose value depends on the other parameters) such that $y^*_M$ is decreasing in $\theta$ for $\theta<\theta_M$ and increasing in $\theta$ for $\theta>\theta_M$. 
\end{theorem}

\begin{proof}
By a similar argument as in the proof of Theorem \ref{thm comparative sharing} we have for all $x\in(\rho,\kappa,\theta,\mu,\beta,\delta,\lambda)$: $\sign(\frac{dy^*_M(r_0)}{dx}) = \sign(G_x(y^*_0,r_0))$ where now $G(y,r)$ is given by
\begin{equation*}
\begin{split}
    G(y,r) &=   1+\left(\rho\left(\frac{1}{2}-(1-\delta) \theta\right)-1-\kappa\right)y-\rho\left(\frac{1}{2}-(1-\delta) \theta\right)y^2\\
    &+\rho y \left((1-\delta) \theta(1-y)\right)^2\frac{\mu-1 }{\beta(y+2(1-y)(1-\delta))}.  
\end{split}
\end{equation*}
We now solve for the sign of each of the partial derivatives of $G$.

\noindent$\boldsymbol{\rho}$: $G_\rho(y,r) = y\left(1-y\right)\left[\frac{1}{2}-\left(1-\delta\right)\theta\left(1-\frac{(\mu-1)(1-y)(1-\delta) \theta}{\beta\left(y+2(1-y)(1-\delta)\right)}\right)\right]>0$.

For the inequality, let $s(y,r)$ denote the expression in square brackets. Then $\sign(G_\rho(y,r)) = \sign(s(y,r))$ and, $s(y,r)=\frac{1}{2}-\left(1-\delta\right)\theta\left(1-a(y,M)\right)>0$, because $1-\delta<\frac{1}{2}$. 

\noindent$\boldsymbol{\theta}$  $G_\theta(y,r) = \rho y\left(1-\delta\right)\left(1-y\right)\left(\frac{2(1-\delta)\theta(1-y)(\mu-1)}{\beta\left(y+2(1-y)(1-\delta)\right)}-1\right)$. 

So, $G_\theta(y,r)>0$ if and only if
\begin{equation*}
\theta>  \frac{\beta\left(y+2(1-y)(1-\delta)\right)}{2(1-\delta)(1-y)(\mu-1)}.
\end{equation*}
Note that the RHS is always positive, so that for sufficiently small $\theta$, $y^*_M$ is decreasing in $\theta$. However, it is possible that the RHS is below 1 so that for large values of $\theta$ the relationship reverses. See Appendix \ref{appendix:additional} for an example.  

\noindent$\boldsymbol{\kappa}$: $G_\kappa(y,r) = -y<0$.

\noindent$\boldsymbol{\mu}$: $G_\mu(y,r) = \frac{ \rho y\left((1-\delta)\theta(1-y)\right)^2}{\beta (y+2(1-y)(1-\delta))}  >0$.

\noindent$\boldsymbol{\beta}$: $G_\beta(y,r) = -\frac{\rho y\left((1-\delta)\theta(1-y)\right)^2(\mu-1) }{\beta ^2 (y+2(1-y)(1-\delta))}<0$.

\noindent$\boldsymbol{\delta}$: $G_\delta(y,r)= -\theta\rho y(1-y)\left[\frac{2(1-\delta)(1-y)(\mu-1) \theta (1-\delta(1-y))}{\beta \left(y+2(1-y)(1-\delta)\right)^2} -1\right]>0$.

For the inequality, let $f(y)$ denote the expression in square brackets. It suffices to prove $f(y)<0$ for all $y$. This follows from $f(0) = \frac{(\mu-1) \theta }{2\beta}-1<0 $ (by Assumption \ref{sufficient for a below 1}), and $f'(y) = -\frac{2(1-\delta)(\mu-1)\theta y }{\beta\left(y+2(1-y)(1-\delta)\right)^3}<0$.

\noindent$\boldsymbol{\lambda}$: $G_\lambda(y,r) = 0$.

\end{proof}
 
\begin{theorem}\label{thm comparative thresholds}

 The thresholds  $\hat{y}_I$ and $\hat{y}_M$ are constant in $\kappa$ and $\rho$ and increasing in $\theta,\mu$, and $\beta$; $\hat{y}_M$ is decreasing in $\delta$ and constant in $\lambda$ and $\hat{y}_I$ is increasing in $\delta$ and decreasing in $\lambda$. 
 \end{theorem}

\begin{proof} For $X\in\{M,I\}$, let $r_0 = (\rho_0,\kappa_0,\theta_0,\mu_0,\beta_0,\delta_0,\lambda_0)$ be a vector of parameters and consider $V(y,X)$ as a function $V^X(y,r):\mathbb{R}^8\rightarrow\mathbb{R}$ (for this proof we use a superscript to distinguish between the two value functions, and subscripts for partial derivatives). Recall that each threshold $\hat{y}^X$ is the unique solution in $(0,1)$ to
\begin{equation}\label{threshold equation}
V^X(\hat{y}_0,r_0)=0.  
\end{equation}
Additionally, recall that by Lemma \ref{lemma thresholds}, for $X\in\{M,I\}$ we have $V^X_y(y,r) >0$ for all $y\in[0,1]$ . Since $V^X(\hat{y}^X_0,r_0)=0$ and $V^X_{y}(\hat{y}^X_0,r_0)\neq 0$, by the implicit function theorem, \eqref{threshold equation} defines a function $\hat{y}^X(r):\mathbb{R}^7\rightarrow \mathbb{R}$ in some neighborhood of $r_0$ and 
\[
    \nabla \hat{y}^X(r_0) = -\frac{1}{V^X_{y}(\hat{y}^X_0,r_0)}\nabla_{r}V^X(\hat{y}^X_0,r_0)
\]

Furthermore, since $V^X_{y}(\hat{y}^X_0,r_0)> 0$, for all $m\in(\rho,\kappa,\theta,\mu,\beta,\delta,\lambda)$, $\sign(\frac{d\hat{y}^X(r_0)}{dm}) = \sign(-V^X_m(\hat{y}^X_0,r_0))$. 
We now use the functional forms of $V(y,M)$ and $V(y,I)$ in \eqref{value functions} to solve for the sign of each of the partial derivatives of $V^X$. First, it is immediate that for $X\in\{M,I\}$ we have $V^X_\kappa(y,r)= V^X_\rho(y,r) =0$. The remaining partial derivatives are presented below. The inequalities hold because by Assumption \ref{sufficient for a below 1},
\[
\frac{(\mu-1-\lambda)(1-y)\delta\theta }{\beta(y+2(1-y)\delta)}< \frac{ (\mu-1)(1-y)(1-\delta)\theta}{\beta(y+2(1-y)(1-\delta))}< \frac{2\beta(1-y)(1-\delta)}{\beta(y+2(1-y)(1-\delta))}<1.
\]
\[
    \begin{split}
V^M_\theta(y,r)&= \frac{ 2(\mu-1)(1-y)(1-\delta)}{y+2(1-y)(1-\delta)} \left( \frac{(\mu-1)(1-y)(1-\delta) \theta}{\beta(y+2(1-y)(1-\delta))}-1\right)<0,\\
          V^I_\theta(y,r) &=  \frac{2(\mu-1-\lambda)(1-y)\delta }{y+2(1-y)\delta}\left( \frac{(\mu-1-\lambda)(1-y)\delta \theta}{\beta(y+2(1-y)\delta)}-1 \right)<0, \\ 
           V^M_\mu(y,r)&=\frac{2(1-y)(1-\delta) \theta}{y+2(1-y)(1-\delta)}\left( \frac{(\mu-1)(1-y)(1-\delta) \theta}{\beta(y+2(1-y)(1-\delta))}-1\right)<0,\\
            V^I_\mu(y,r) &=\frac{2(1-y)\delta \theta}{y+2(1-y)\delta}\left(\frac{(\mu-1-\lambda)(1-y)\delta \theta}{\beta(y+2(1-y)\delta)}-1 \right)<0,\\
       V^M_\beta(y,r) &=-\frac{1}{\beta^2}\left(\frac{(\mu-1)(1-y)(1-\delta) \theta}{y+2(1-y)(1-\delta)}\right)^2<0,\\
             V^I_\beta(y,r)&= -\frac{1}{\beta^2}\left(\frac{(\mu-1-\lambda)(1-y)\delta \theta}{y+2(1-y)\delta}\right)^2 <0,\\
                V^M_\delta(y,r) &=\frac{2(1-y)y}{\left(y+2(1-y)(1-\delta)\right)^2}\left(1+\theta(\mu-1)\left(1-\frac{(\mu-1)(1-y)(1-\delta)\theta}{\beta \left(y+2(1-y)(1-\delta)\right)}\right)\right)>0, \\ 
                V^I_\delta(y,r) & = \frac{2 (1-y) y}{\left(y+2(1-y)\delta\right)^2}\left(\theta(\mu-1-\lambda)\left(\frac{(\mu-1-\lambda)(1-y)\delta\theta}{\beta \left(y+2(1-y)\delta\right)}-1\right)-1-\lambda \right)<0,\\ 
                 V^M_\lambda(y,r) & = 0, \\
                   V^I_\lambda(y,r) & = \frac{y}{y+2(1-y)\delta} +\frac{2(1-y)\delta \theta}{y+2(1-y)\delta}\left(1-\frac{(\mu-1-\lambda)(1-y)\delta \theta}{\beta\left(y+2(1-y)\delta\right)}\right)>0.  
    \end{split}
\]
\end{proof}
 
\begin{theorem}\label{thm:rho limit}
Fix values for the other parameters and let $y^*_R(\rho)$ be the quasi steady state for region $R\in\{I,M,S\}$  as a function of $\rho$, and $y^\infty_R := \lim\limits_{\rho \to \infty} y^*_R(\rho)$. This limit exists for all regions $R$ and:\footnote{In Appendix \ref{appendix:additional} we show that all cases described in statements 2-4 are possible, i.e., $y^*_I(\rho)$ can converge to either $0,1$ or to an interior point, and when the limit is interior it can be either increasing or decreasing towards its limit.}
\begin{enumerate}
    \item $y^\infty_S =  y^\infty_M = 1$.
    \item If $d_I(y)>0$ ($<0$) for all $y \in (0,1)$, then 
$ y^\infty_I = 1$ ($y^\infty_I = 0$). 
    \item If there exists $\bar y \in (0,1)$ such that $d_I(\bar y)=0$, then 
$y^\infty_I = \bar y$.

\item If $d_I(\frac{1}{1+\kappa})>0$ ($<0$), then $y^*_I(\rho)$ is strictly increasing (decreasing) toward its limit.
\end{enumerate}
\end{theorem}

\begin{proof}
The following observations hold for all $R\in\{I,M,S\}$: By \eqref{general ODE},
    \begin{equation}\label{eqn:y rho}
        y^*_R(\rho) = \frac{1+p^T_R(y^*_R(\rho))\rho}{1+\kappa+\rho(p^T_R(y^*_R(\rho))+p^F_R(y^*_R(\rho)))}.
    \end{equation} 
and by \eqref{probabilities}, $p^T_R(y)+p^F_R(y)>0$ for all $y\in(0,1)$. So if the limit $y^\infty_R$ exists and is in $(0,1)$,  then $ y^\infty_R= \frac{p^T_R(y^\infty_R)}{p^T_R(y^\infty_R)+p^F_R(y^\infty_R)}$, which simplifies to $d_R(y^\infty_R)=0$. Thus, if the limit defining $y^\infty_R$ exists then either $y^\infty_R\in\{0,1\}$ or $y^\infty_R\in(0,1)$ and satisfies $d_R(y^\infty_R)=0$. Additionally, by \eqref{eqn:y rho}, $y^*_R(0)=\frac{1}{1+\kappa}$. Next,

\begin{enumerate}
    \item By Theorems \ref{thm comparative sharing} and \ref{thm comparative boring}, both $y^*_I$ and $y^*_S$ are monotonically increasing in $\rho$, so the limit exists for both of these decision rules. As explained in the main text, for $R=M,S$ users are always discerning, i.e., $d_R(y)>0$ for all $y\in(0,1)$. So for both of these decision rules $y^\infty_R>y^*_R(0)>0$ and $y^\infty_R\in\{0,1\}$, which implies $ y^\infty_S =  y^\infty_M = 1$.
    \item By Theorem \ref{thm comparative evoc}, if $d_I(y)>0$ ($d_I(y)<0$) for all $y$ then $y^*_I(\rho)$ is monotone increasing (decreasing). Thus, in both of the cases the limit exists, and it  cannot be interior since there does not exist an interior $\bar y$ with $d_I(\bar y)=0$. Now, if $d_I(y)>0$ for all $y$ then $y^\infty_I>y^*_I(0)>0$ so $y^\infty_I=1$. If $d_I(y)<0$ for all $y$ then $y^\infty_I<y^*_I(0)<1$ so $y^\infty_I=0$.

    \item Assume that there exists  $\bar y \in (0,1)$ such that $d_I(\bar y)=0$. We will prove that $y^*_I(\rho)$ converges to $\bar y$, and that $y^*_I(\rho)$ is monotone increasing (decreasing) if $d_I(\frac{1}{1+\kappa})>0$ $(<0)$.  First, consider the case where $d_I(\frac{1}{1+\kappa})>0$. Since attention is decreasing in $y$, $d_I(y) = \frac{1}{2}-\delta\theta \left(1-a(y,I)\right)$ is also decreasing in $y$. By Theorem \ref{thm comparative evoc},  $y^*_I(\rho)$ is increasing at all $\rho$ for which $d_I(y^*_I(\rho))>0$, and decreasing when the inequality is reversed. Since $y^*_I(0)=\frac{1}{1+\kappa}$ the assumption $d_I(\frac{1}{1+\kappa})>0$ implies that $y^*_I(\rho)$ is bounded above by $\bar y$ and increasing for all $\rho$. Hence, the limit defining $y^\infty_R$  exists and is equal to $\bar y$. A symmetric argument shows that if $d_I(\frac{1}{1+\kappa})<0$ then $y^*_I(\rho)$ is strictly decreasing and converges to $\bar y$, which completes the proof of claims 3 and 4.  
    
\end{enumerate}
\end{proof}

\begin{proof}[Proof of Theorem \ref{thm: comparative kappa}]
     By Lemma \ref{ordering lemma}, we have $\min\{y^*_S,y^*_M\}>y^*_N=\frac{1}{1+\kappa}$. Thus, $y^*_S,y^*_M$, and $y^*_N$ all converge to $1$ as $\kappa$ goes to $0$. Since the thresholds are constant in $\kappa$, this implies that there exists some $\kappa'>0$ such that $y^*_S,y^*_M,y^*_N$ are all in region $S$ for all $\kappa<\kappa'$. Hence, $y^*_S\in\mathcal{S}_\kappa$, and $y^*_M,y^*_N\notin\mathcal{S}_\kappa$ for all $\kappa<\kappa'$. If the intermediate region is $M$, the above implies that $\mathcal{S}_\kappa=\{y^*_S\}$ for all $\kappa<\kappa'$ (the only other possible stable steady state in this case is $\hat{y}_I$, which cannot be stable if $y^*_S$ is in region $S$). Consider the case where the intermediate region is $I$. Let $y^*_I(\kappa)$ denote the quasi steady state $y^*_I$ as a function of $\kappa$. Let $y^*_I(0):=\lim_{\kappa\rightarrow 0} y^*_I(\kappa)$. If $y^*_I(0)\in I$, then there exists $\kappa_1>0$ such that $\mathcal{S}_\kappa=\{y^*_S,y^*_I\}$ for all $\kappa<\kappa_1$. If  $y^*_I(0)\in N$, then there exists $\kappa_1>0$ such that $\mathcal{S}_\kappa=\{y^*_S,\hat y_I\}$ for all $\kappa<\kappa_1$.Finally, if $y^*_I(0)\in S$, then there exists $\kappa_1>0$ such that $\mathcal{S}_\kappa=\{y^*_S\}$ for all $\kappa<\kappa_1$. Numerical examples in Appendix \ref{appendix:additional} show that all three cases are possible. This completes the proof of the first claim, where we use our standing assumption that no two variables in $\{\hat{y}_I,\hat{y}_M,y^*_N,y^*_I,y^*_S\}$ are equal.
     To prove the second claim, note that as $\kappa\rightarrow\infty$, all quasi steady states converge to $0$ (since $y^*_N=1/(1+\kappa)\rightarrow 0$ and by \eqref{y_R inequality} all others are  bounded above by $\frac{1+\rho}{1+\kappa+\rho}$), while the thresholds $\hat{y}_I,\hat{y}_M$ are constant in $\kappa$. Hence there exists $\kappa_2$ large enough that all quasi steady states fall below $\min\{\hat{y}_I,\hat{y}_M\}$ and therefore lie in region $N$, giving $\mathcal{S}_\kappa=\{y^*_N\}$ for all $\kappa>\kappa_2$. \end{proof}

\subsubsection{Summary}

We summarize the takeaways from the comparative statics results above that are not discussed in the main text.

The candidate limit point $y^*_N=\frac{1}{1+\kappa}$ is decreasing in $\kappa$ and constant in all other parameters. As one would expect, the other candidate limit points are increasing in the loss $\mu$  from  sharing false stories. Additionally, any limit point that is a quasi steady state is decreasing in the exogenous inflow of false stories $\kappa$ and all but $y^*_N$ are  decreasing in the cost of attention  $\beta$. 

More surprisingly, the candidate limit point  $\hat{y}_I$ is increasing in $\beta$ and constant in $\kappa$. Recall that $\hat{y}_I$ is the point where users are exactly indifferent between sharing and not sharing very interesting stories. It is increasing in $\beta$ because users' payoffs are decreasing in the cost of attention and increasing in the share of true stories. Hence, when $\beta$ goes up, the share of true stories required for indifference needs to go up as well to compensate for the utility loss. $\hat{y}_I$ does not depend on $\kappa$, since the exogenous inflow of false stories is not an argument in users' utility functions.

\begin{table}[h]
\centering
\begin{minipage}[t]{0.45\textwidth}
\centering
\caption{\textbf{Comparative Statics for $\kappa$}}
\label{table:kappa_beta}
\setlength{\tabcolsep}{12pt} 
\begin{tabular}{@{\extracolsep{\fill}}ll@{}}
\toprule
\(y^*_M, y^*_S, y^*_I\) & Decreasing. \\
\(\hat{y}_I\) & Constant. \\
\bottomrule
\end{tabular}
\end{minipage}
\hfill
\begin{minipage}[t]{0.45\textwidth}
\centering
\caption{\textbf{Comparative Statics for $\beta$}}
\label{table:beta}
\setlength{\tabcolsep}{12pt} 
\begin{tabular}{@{\extracolsep{\fill}}ll@{}}
\toprule
\(y^*_M, y^*_S, y^*_I\) & Decreasing. \\
\(\hat{y}_I\) & Increasing. \\
\bottomrule
\end{tabular}
\end{minipage}
\end{table}

The quasi steady states $y^*_N$ and $y^*_M$ are constant in the interest-level parameter  $\lambda$; all other candidate limit points are decreasing in $\lambda$. As $\lambda$ increases, users pay less attention to the veracity of very interesting stories. This leads to a decrease in the share of true stories in any limit point where users share very interesting stories, which is the case when $\lambda$ is sufficiently large. Comparative statics with respect to $\delta$ are more nuanced. 

  \medskip

  \bigskip

  \noindent\textbf{The role of $\delta$}
  \begin{table}[h]
  \centering
  \caption{  \textbf{Comparative Statics for $\delta$}}
  \label{table:delta}
  \setlength{\tabcolsep}{10pt} 
  \begin{tabular}{ll@{}}
  \toprule
  \(y^*_M\) & Increasing. \\
  \(y^*_S\) & Decreasing for $\delta$ close to $\frac{1}{2}$, and increasing for $\delta$ close to $1$. \\
  \(y^*_I\) & Decreasing. \\
  \(\hat{y}_I\) & Increasing. \\
  \bottomrule
  \end{tabular}
  \end{table}

 Increasing $\delta$ means false stories are more likely to be very interesting, so the comparative statics for $y^*_I,y^*_M$ are intuitive---the limit share of true stories decreases (increases) in $\delta$ when users share only very interesting (mildly interesting) stories. The quasi steady state $y^*_S$, where users share both types of stories, decreases in $\delta$ when $\delta$ is close to $\frac{1}{2}$, and increases in $\delta$ when $\delta$ is close to $1$.   Appendix \ref{appendix:online}  presents numerical examples where $y^*_S$ is both  decreasing and increasing in $\delta$ when it is a limit point. Intuitively, the non-monotonicity arises because when $\delta$ is close to $\frac{1}{2}$ users are sharing more very interesting stories than mildly interesting stories, since both types of stories are almost equally likely to be false and very interesting stories have additional value. In this case, the comparative statics with respect to $\delta$ are similar to those comparative statics for $y^*_I$, where  users are only sharing very interesting stories. As $\delta$ moves closer to $1$, the stories that users share are more likely to be mildly interesting and comparative statics with respect to $\delta$ eventually become similar to those for $y^*_M$. Finally, $\hat{y}_I$ is increasing in $\delta$ because for a fixed $y_n$, increasing $\delta$ leads to a decrease in the value from sharing very interesting stories.\footnote{This can lead to a counter-intuitive situation where asymptotically users only share very interesting stories, but when very interesting stories become more likely to be false the limit share of true stories increases. This happens when $\hat{y}_I$ is a limit point and is between regions $N$ and $I$ so users are mixing between sharing very interesting stories and not sharing.} 

\bigskip

 \subsection{Additional Details}\label{appendix:additional}
Differentiation of the functions  $a(y,I)$ and $a(y,M)$ shows that
\[
\begin{split}
   \frac{\partial a(y,M)}{\partial y} & = -\dfrac{ (\mu -1)(1-\delta) \theta }{\beta\left(y+2(1-y)(1-\delta)\right)^2}<0, \\ 
    \frac{\partial a(y,I)}{\partial y} &= -\dfrac{  (\mu-\lambda - 1)\delta  \theta}{\beta\left(y+2(1-y)\delta\right)^2}<0, \\
    \frac{\partial a(y,M)}{\partial \theta} & =\dfrac{(\mu-1)(1-y)(1-\delta)}{\beta\left(y+2(1-y)(1-\delta)\right)} >0,\\
       \frac{\partial a(y,I)}{\partial \theta}&= \dfrac{(\mu-1-\lambda)(1-y)\delta}{\beta\left(y+2(1-y)\delta\right)} >0,\\
      \frac{\partial a(y,M)}{\partial \delta} &= -\dfrac{  (\mu -1) (1-y) y\theta}{\beta\left(y+2(1-y)(1-\delta)\right)^2}<0,\\
        \frac{\partial a(y,I)}{\partial \delta} &= \dfrac{  (\mu-1-\lambda)(1-y) y\theta}{\beta\left(y+2(1-y)\delta\right)^2}>0,\\
        \frac{\partial a(y,M)}{\partial \beta} &=-\dfrac{(\mu-1)(1-y)(1-\delta) \theta}{\beta^2\left(y+2(1-y)(1-\delta)\right)} <0,\\
         \frac{\partial a(y,I)}{\partial \beta} & =  -\dfrac{(\mu-1-\lambda)(1-y)\delta \theta}{\beta^2\left(y+2(1-y)\delta\right)}<0,
         \end{split}
        \]
               \[
    \begin{split}
          \frac{\partial a(y,M)}{\partial \lambda} & =  0,\\
            \frac{\partial a(y,I)}{\partial \lambda} & = -\dfrac{(1-y)\delta \theta}{\beta\left(y+2(1-y)\delta\right)} <0,\\
              \frac{\partial a(y,M)}{\partial \mu} & =  \dfrac{(1-y)(1-\delta) \theta}{\beta\left(y+2(1-y)(1-\delta)\right)} >0,\\
               \frac{\partial a(y,I)}{\partial \mu} & = \dfrac{(1-y)\delta \theta}{\beta\left(y+2(1-y)\delta\right)} >0,
    \end{split}
\]
where we use Assumption \ref{condition on lambda} to sign the partial derivatives.

Below we present numerical examples for claims made in main text. All examples satisfy our standing parametric assumptions, i.e., all parameters are strictly positive, satisfy Assumptions \ref{condition on lambda} and \ref{sufficient for a below 1}, and $\theta<1,\delta\in(\frac{1}{2},1)$. 
\subsubsection*{Numerical Examples for Sections \ref{subsec:sharing decision} and \ref{subsec:dynamics}}

For a numerical example that the relationships between $y^*_S$ and $y^*_M$ and between $y^*_I$ and $y^*_N$ can go both ways fix $\beta = \kappa = \rho = 1$, $\mu=1.75$, $\lambda=0.25$, and $\theta=0.75$. Calculations show that  $y^*_M<y^*_S$ for $\delta \lessapprox 0.745$ and $y^*_M>y^*_S$ for $\delta \gtrapprox 0.745$. Additionally, $y^*_N<y^*_I$ for $\delta  \lessapprox0.751$ and $y^*_N>y^*_I$ for $\delta\gtrapprox0.751$. Thus, Lemma \ref{ordering lemma} is ``all we can know" regarding the ordering of the quasi steady states. Likewise, the relationship between the thresholds $\hat{y}_I,\hat{y}_M$ is undetermined. Calculations with the same parameter values as above show that $\hat{y}_I<\hat{y}_M$ for $\delta \lessapprox0.647$ and  $\hat{y}_I>\hat{y}_M$ for $\delta\gtrapprox0.647$. Finally, the relationship between the thresholds and quasi steady states is also not determined: Both $\max\{y^*_I,y^*_M,y^*_N,y^*_S\}<\min\{\hat{y}_I,\hat{y}_M\}$ and $\min\{y^*_I,y^*_M,y^*_N,y^*_S\}>\max\{\hat{y}_I,\hat{y}_M\}$ are possible (see the numerical examples for Theorem \ref{thm: comparative kappa}).

We now show that both of the configurations that give rise to Case (a) of Theorem \ref{thm S_F} are possible. For an example where $\hat{y}_I<\hat{y}_M$ and $y^*_I<\hat{y}_I<y^*_N$, set $\rho=20,\theta=0.9,\kappa=12,\mu=1.55,\beta=1,\delta=0.65,\lambda = 0.45$. For an example where $\hat{y}_I>\hat{y}_M$ and $y^*_S<\hat{y}_I<y^*_M$, set $\rho=1,\theta=0.9,\kappa=2.55,\mu=1.65, \beta=1,\delta=0.8,\lambda=0.25$. It can be verified that in both of these examples $\hat{y}_I$ is the unique stable steady state of the LDI. 

\subsubsection*{Numerical Examples for Sections \ref{subsec:comparative} and \ref{appendix:comparative}}

\paragraph{Non-monotonicity in $\theta$:}

We now show that each quasi steady state  $y^*_M,y^*_S,y^*_I$ can be first decreasing and then increasing in $\theta$ when it is a steady state for the LDI (and thus a limit point for the system). For $y^*_S$, set  $\rho=0.3, \kappa=1.5,\mu=1.57,\beta=0.3,\delta=0.55,\lambda = 0.05$.  With these parameters, $y^*_S$ is in region $S$ for all $\theta\in(0,1)$ and is decreasing in $\theta$ for $\theta\lessapprox 0.95$ and then increasing. Furthermore, for all values of $\theta\in(0,1)$, all other candidate limit points are also in region $S$, so $y^*_S$ is the unique limit point of the system for any value of $\theta$. 

For $y^*_M$, set $\rho=1, \kappa=8,\mu=1.57,\beta=0.3,\delta=0.9,\lambda = 0.05$. With these parameters, $\hat{y}_M<\hat{y}_I$ for all $\theta\in(0,1)$, so the intermediate region is $M$. Additionally, $y^*_M$ is in region $M$ for all $\theta \gtrapprox 0.16$ (otherwise, $y^*_M$ is in region $S$), and $y^*_M$ is decreasing in $\theta$ for $\theta \lessapprox0.87$ and then increasing in $\theta$. So $y^*_M$ is both decreasing and increasing in $\theta$ in region $M$.

Finally, for $y^*_I$, set  $\rho=0.45,\kappa=3.34,\mu=1.54,\beta=0.3,\delta=0.53,\lambda=0.1$. With these parameters, $\hat{y}_I<\hat{y}_M$ for all $\theta\in(0,1)$, so the intermediate region is $I$. Additionally, $y^*_I$ is in region $I$ for all $\theta\gtrapprox0.85$ (in region $S$ for smaller $\theta$), and $y^*_I$ is decreasing in $\theta$ for $\theta\lessapprox0.88$ and then increasing in $\theta$. So $y^*_I$ is non-monotone in $\theta$ in region $I$.  
 
\paragraph{Non monotonicity in $\delta$:}

For an example that $y^*_S$  can decrease and then increase in $\delta$ when it is a steady state for the LDI, again set $\beta = \kappa = \rho = 1$, $\mu=1.75$, $\lambda=0.25$, and $\theta=0.75$. With these parameters, $y^*_S>\max\{\hat{y}_I,\hat{y}_M\}$ for all $\delta\in(\frac{1}{2},1)$ so that $y^*_S$ is a steady state for the LDI for any value of $\delta$. Additionally, $y^*_S$ is decreasing in $\delta$ for $\delta\lessapprox 0.73$ and increasing in $\delta$ for $\delta\gtrapprox 0.73$. 

\paragraph{Dependence of $y^*_I$ on $\rho$:}

We now show that $d_I(y^*_I)$ can be either positive or negative, which (by Theorem \ref{thm comparative evoc}) implies that $y^*_I$ can be either increasing or decreasing in $\rho$, and that both cases can occur when $y^*_I$ is a limit point. We also show that all cases described by Theorem \ref{thm:rho limit} are possible: For fixed values of the other parameters $y^*_I(\rho)$ can converge to either $0,1$ or to an interior point, and when the limit is interior it can be either increasing or decreasing towards its limit. 

In the following numerical examples we fix values for all parameters except $\rho$ and $\delta$ and consider comparative statics with respect  to $\rho$ at four different values of $\delta$. In the first two specifications, discernment with decision rule $I$ is negative, so $y^*_I$ is decreasing in $\rho$. In the third and fourth specifications discernment is  positive so $y^*_I$ increases in $\rho$.

Set $\theta=0.9,\kappa=3,\mu=1.55,\beta=1, \lambda =0.45$, and $\delta = 0.8$. With these parameters $\hat{y}_I<\hat{y}_M$, so the intermediate region is $I$ (for any value of $\rho$). For these parameters, $d_I(y)<0$ for all $y\in(0,1)$ so, by Theorem \ref{thm:rho limit}, $y^*_I(\rho)\rightarrow 0$. Starting with $\rho=0$, we have $y^*_I$ in region $S$, and $y^*_I$ is decreasing in $\rho$ such that it enters region $I$ when $\rho\approx 15.6$, and enters region $N$ when $\rho\approx 42.2$ (so it is a limit point when $\rho$ is between those values). 

With the same parameter values but setting $\delta=0.576$, the intermediate region is again $I$, but $y^*_I\in I$ for $\rho=0$. It is still the case that $d_I(y^*_I(0))=d_I(\frac{1}{1+\kappa})<0$,  so $y^*_I$ is still decreasing in $\rho$, but now for $\bar y \approx 0.24 $ we have $d_I(\bar y)=0$ so $y^*_I(\rho)\rightarrow \bar y$ and is decreasing towards its limit. In this case, $\bar y\in I$ so $y^*_I$ converges to an interior point and remains in region $I$ for any value of $\rho$. Note that in this specification and the previous one $d_I(y^*_I(\rho))<0$ for all $\rho$. We now present two examples where this inequality is reversed.

If we further decrease $\delta$ to $\delta=0.57$, the intermediate region is again $I$, and $y^*_I\in I$ for $\rho=0$. However, now $d_I(\frac{1}{1+\kappa})>0$ so $y^*_I$ is increasing in $\rho$ and enters region $S$ when $\rho\approx 142.7$. For these parameters $d_I(\tilde y)=0$ for $\tilde y\approx 0.47$ so $y^*_I(\rho)\rightarrow \tilde y$ and is increasing towards its limit. Finally, setting $\delta=0.55$, we get that $d_I(y)>0$ for all $y\in(0,1)$, so that $y^*_I(\rho)\rightarrow 1$. 

These examples demonstrate that making false stories less likely to be very interesting (decreasing $\delta$) can reverse the effect of increasing reach, and that small changes in $\delta$ can have large effects on the limit of $y^*_I$ in $\rho$.

\paragraph{Examples for Theorem \ref{thm: comparative kappa}} We demonstrate that when the intermediate region is $I$, for  $y^*_I(0):=\lim_{\kappa\rightarrow 0} y^*_I(\kappa)$, all three of the following cases are possible: $y^*_I(0)\in N$, $y^*_I(0)\in I$, and $y^*_I(0)\in S$. As explained in the proof of Theorem \ref{thm: comparative kappa}, each of these cases leads to a different composition of $\mathcal{S}_\kappa$ for small values of $\kappa$.

Fix all parameters except $\kappa$ at $\rho=100,\theta=0.9,\mu=1.55,\beta=1,\delta=0.8,\lambda=0.45$. With these parameters, the intermediate region is $I$. For all $\kappa>0$, we have $y^*_I\in N$, and for all $\kappa\lessapprox 10.2$ we have $y^*_N>\hat{y}_I$. Thus, for all $\kappa\in (0,10.2)$ we have $\hat{y}_I\in\mathcal{S}_\kappa$. 

Next, set $\lambda=0.52$ and leave the remaining parameters as above. The intermediate region is still $I$, but now for all $\kappa\lessapprox 14.9$ we have $y^*_I\in I$. Thus, for all $\kappa\in (0,14.9)$ we have $y^*_I\in\mathcal{S}_\kappa$.

Finally, take the configuration in the previous paragraph but set $\rho=5$. The intermediate region is still $I$ but now for all $\kappa\lessapprox 4.6$ we have $y^*_I\in S$. Thus, for all $\kappa\in (0,4.6)$ we have $y^*_I\notin\mathcal{S}_\kappa$ and $\hat{y}_I\notin \mathcal{S}_\kappa$.

\subsection{Phase Diagrams}\label{appendix:phase diagrams}

As mentioned in the main text,  five variables  pin down the phase diagram of the continuous-time system: the two thresholds $\hat{y}_I$ and $\hat{y}_M$, and the quasi steady states for the system's three regions, i.e., $y^*_S,y^*_N$ and one of $y^*_I,y^*_M$. In this section, we explain why there are 40 configurations of these variables that satisfy the model's restrictions, and present phase diagrams for all possible configurations. Recall that the only restriction implied by the analysis so far is Lemma \ref{ordering lemma}, i.e.,  $\min \{y^*_S,y^*_M\}> \max\{y^*_I,y^*_N\}$, and that numerical computations show that any configuration of the variables that satisfies this restriction arises for some parameter values.

To see why there are 40 configurations, consider the case  $\hat{y}_I<\hat{y}_M$. In this case, the five variables are $\{\hat{y}_I,\hat{y}_M,y^*_N,y^*_I,y^*_S\}$. We can now count the number of orderings of these variables that satisfy our restrictions. First, we can choose the relative positions of the two thresholds, giving ${5\choose 2} = 10$ options. Lemma \ref{ordering lemma} shows that $y^*_S>\max\{y^*_N,y^*_I\}$, and  $\hat{y}_I<\hat{y}_M$ by assumption, so the only degree of freedom is the order between $y^*_N,y^*_I$, for a total of $20$ configurations in which $\hat{y}_I<\hat{y}_M$. Similarly, there are 20 configurations with $\hat{y}_I>\hat{y}_M$.

 Figures \ref{phasediagram_1}, \ref{phasediagram_2}, \ref{phasediagram_3} and \ref{phasediagram_4} present phase diagrams for all possible configurations of the thresholds and quasi steady states. Stable steady states are in green, repelling steady states are in red, and quasi steady states that are not steady states are in purple. The numbers on the bottom left of each phase diagram are the indices of the positions of the two thresholds among the five variables that pin down the phase diagram. For example, in the phase diagram on the top left of each figure, the thresholds are in the first and second positions. 
 
 \bigskip
 
\begin{figure}[]
    \centering
    \includegraphics[width=0.75\textwidth]{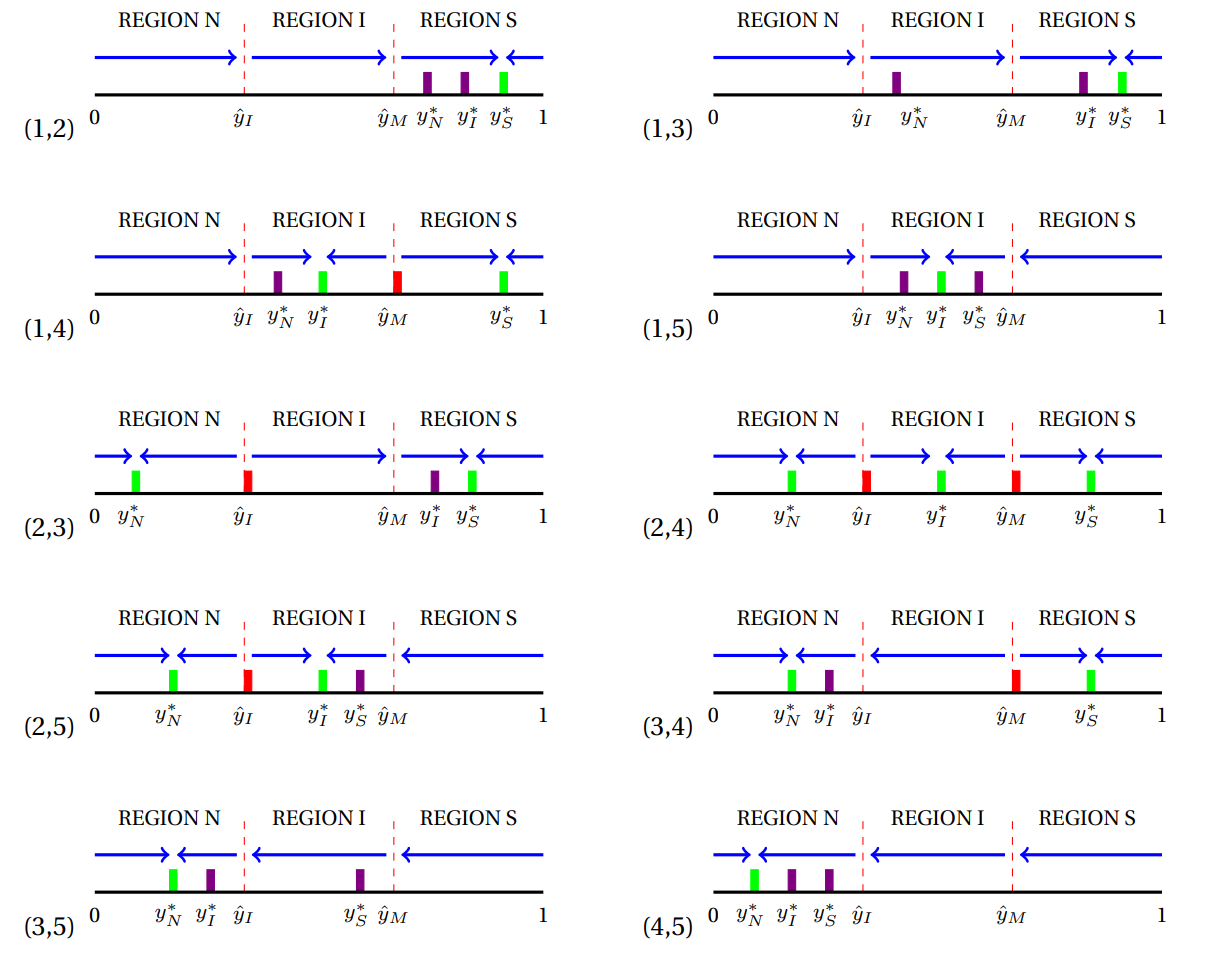}
    \caption{Phase diagrams for the case $\hat{y}_I<\hat{y}_M; y^*_I>y^*_N$.}
    \label{phasediagram_1}
\end{figure}

\begin{figure}[]
    \centering
    \includegraphics[width=0.75\textwidth]{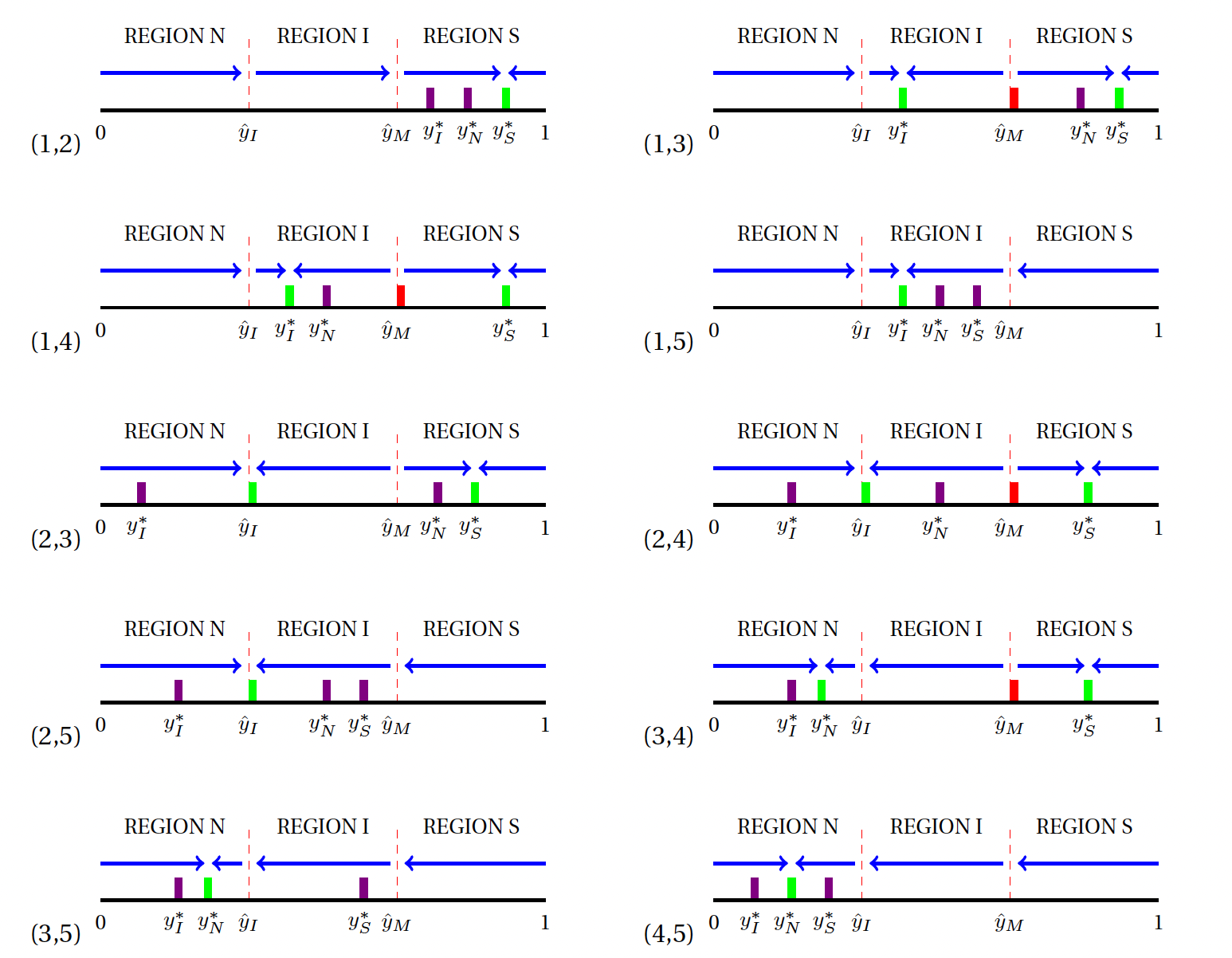}
    \caption{Phase diagrams for the case $\hat{y}_I<\hat{y}_M; y^*_I<y^*_N$.}
    \label{phasediagram_2}
\end{figure}

\begin{figure}[]
    \centering
    \includegraphics[width=0.75\textwidth]{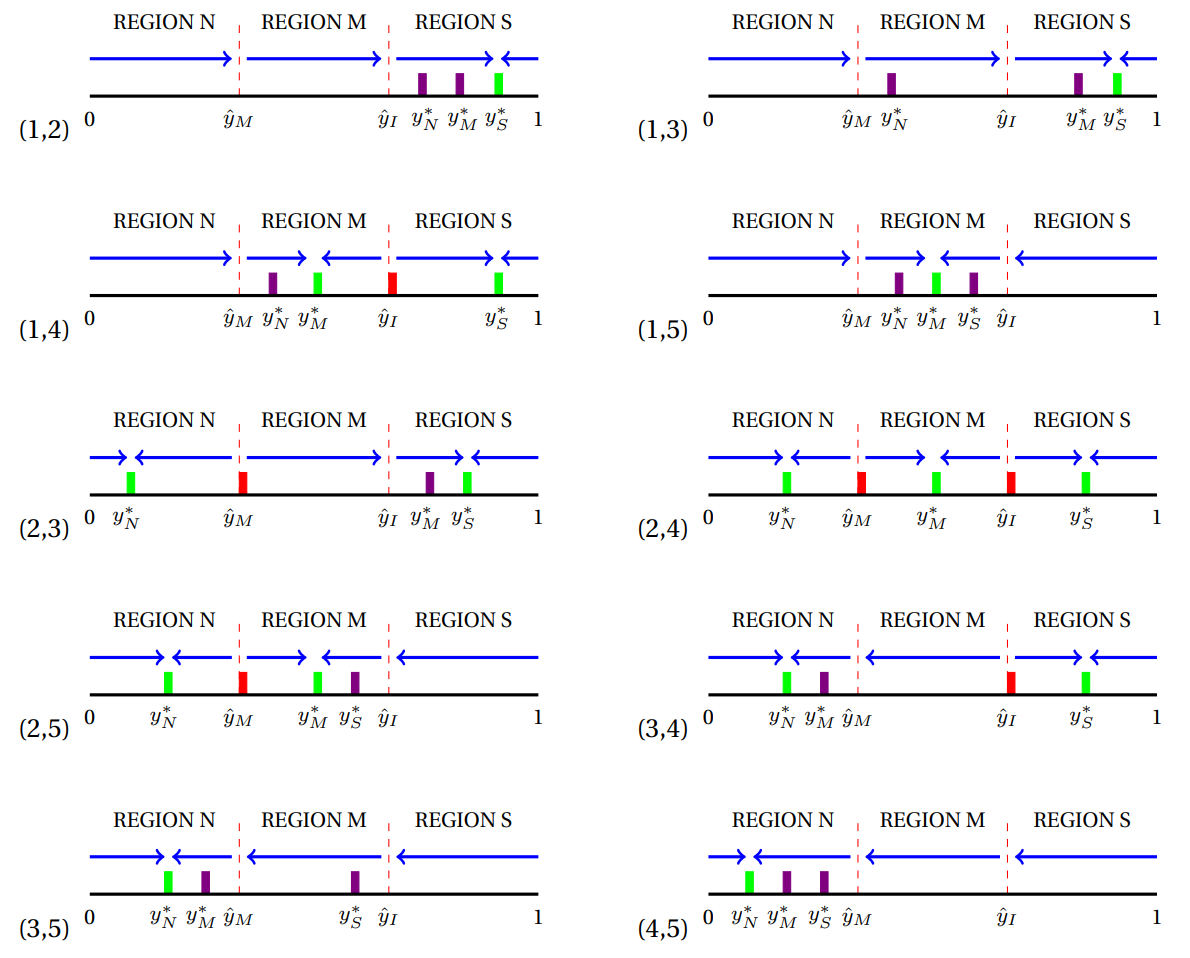}
    \caption{Phase diagrams for the case $\hat{y}_I>\hat{y}_M; y^*_S>y^*_M$.}
    \label{phasediagram_3}
\end{figure}

\begin{figure}[]
    \centering
    \includegraphics[width=0.75\textwidth]{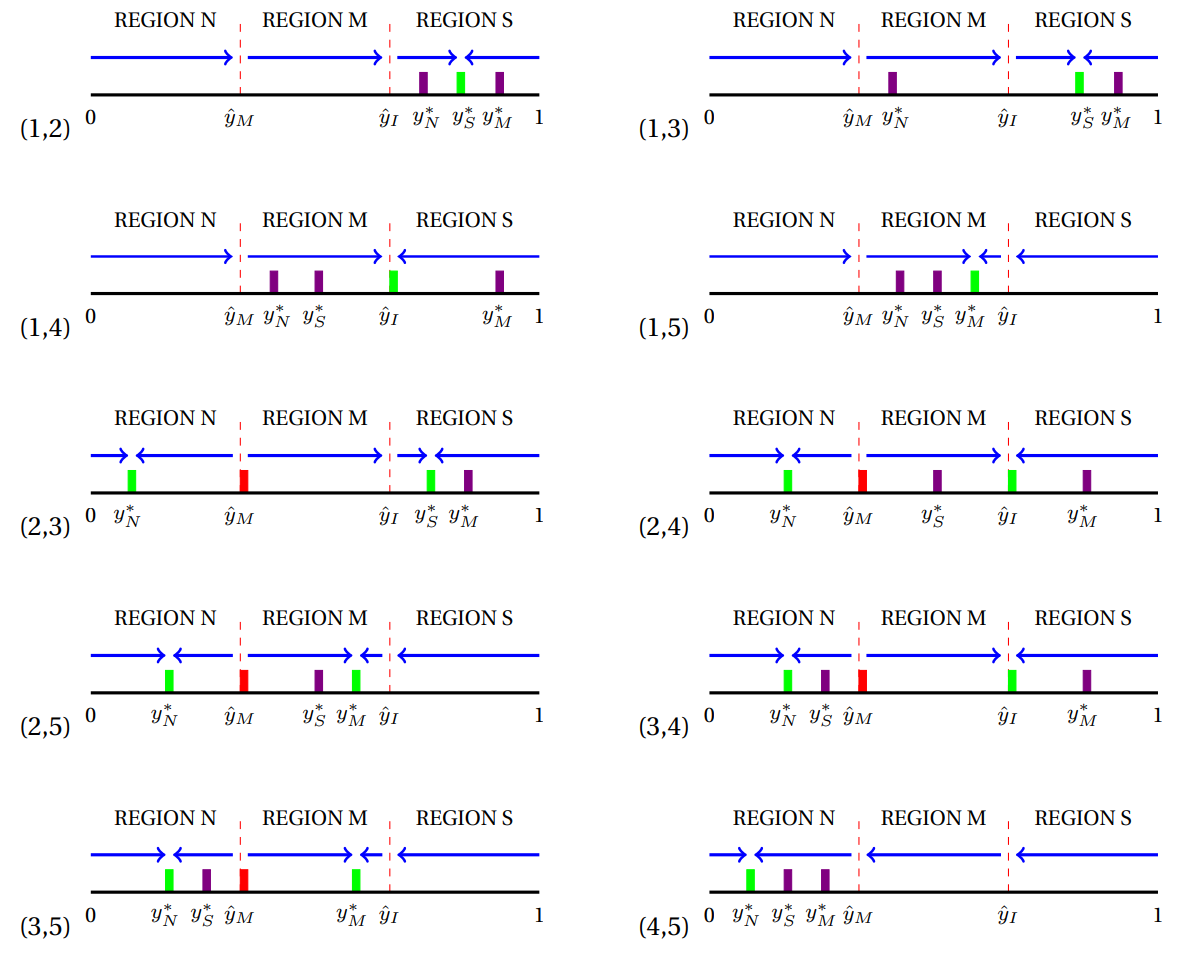}
    \caption{Phase diagrams for the case $\hat{y}_I>\hat{y}_M; y^*_S<y^*_M$.}
    \label{phasediagram_4}
\end{figure}

\end{appendix}
\end{document}